\documentclass[journal,draftcls,onecolumn,12pt,twoside]{IEEEtranTCOM}
\normalsize

\ifCLASSINFOpdf
   \usepackage[pdftex]{graphicx}
\else
\fi

\usepackage{amsmath,amssymb,amsfonts,amsthm}
\usepackage{quotmark}

\newtheorem{lemma}{Lemma}
\newtheorem{theorem}{Theorem}

\newtheorem{example}{Example}

\newtheorem{note*}{Note}

\usepackage[ruled]{algorithm2e}
\usepackage{makecell}
\usepackage{bm}
\usepackage{cancel}

\usepackage{subcaption}

\begin{document}
%
\title{On Cache-Aided Multi-User Private Information Retrieval with Small Caches}
%
%
%

\author{Charul~Rajput~
      and~B. Sundar~Rajan \\ \small{Department of Electrical Communication Engineering, Indian Institute of Science Bengaluru, India \\
E-mail: charulrajput@iisc.ac.in, bsrajan@iisc.ac.in}}
\maketitle

\begin{abstract}
In this paper, we propose a scheme for the problem of cache-aided multi-user private information retrieval with small caches, in which $K$ users are connected to $S$ non-colluding databases via shared links. Each database contains a set of $N$ files, and each user has a dedicated cache of size equivalent to the size of $M$ files. All the users want to retrieve a file without revealing their demands to the databases. During off-peak hours, all the users will fill their caches, and when required, users will demand their desired files by cooperatively generating query sets for each database. After receiving the transmissions from databases, all the users should get their desired files using transmitted data and their cache contents. This problem has been studied in [X. Zhang, K. Wan, H. Sun, M. Ji and G. Caire, \tqt{Fundamental limits of cache-aided multiuser private information retrieval}, IEEE Trans. Commun., 2021], in which authors proposed a product design scheme.
In this paper, we propose a scheme that gives a better rate for a particular value of $M$ than the product design scheme. We consider a slightly different approach for the placement phase. Instead of a database filling the caches of all users directly, a database will broadcast cache content for all users on a shared link, and then the users will decide unitedly which part of the broadcasted content will be stored in the cache of each user. This variation facilitates maintaining the privacy constraint at a reduced rate.
\end{abstract}

\begin{IEEEkeywords}
Private information retrieval, coded caching, dedicated caches.
\end{IEEEkeywords}

%
\IEEEpeerreviewmaketitle

\section{Introduction}
%
%
%
%

In a problem of private information retrieval (PIR), a user wants to retrieve a file from distributed databases containing the same set of files in such a way that the identity of the desired file is private. This problem was first introduced by Chor et al. \cite{CKGS1998}, and since then it has been studied extensively in literature \cite{SJ2017A, SJ2017B, T2017, ZWSJC2021}. The trivial solution for the PIR problem is to send all the files to the user to keep the demand identity private, but clearly, it is not a feasible solution as the number of files can be large. Hence the requirement is to find a PIR scheme that minimizes the communication overhead that occurred due to the privacy constraint. 

There are many variations of the PIR problem in the literature, such as multi-message PIR (MPIR), in which the user can demand more than one message \cite{BU2018}, PIR with side information (PIR-SI), in which the user has prior side information about the database \cite{KGHES2019},
pliable PIR (PPIR) in which the user is interested in any message from a desired subset of the set of all messages stored in databases \cite{OK2022}. One variant of PIR is studied with coded caching \cite{KYL2017}, in which the user is equipped with a cache, and during off-peak hours a database place some content in the cache of the user. When the user wants to retrieve a file, then the transmissions will be made in such a way that the user can get the desired file using transmitted data and cache content while the privacy of the demand is maintained. After that, the problem cache-aided PIR has been studied for multiple users equipped with dedicated caches and is called cache-aided multi-user PIR (MuPIR) \cite{ZWSJC2021}. Recently, in \cite{VR2022}, authors have proposed a scheme for multi-access cache-aided multi-user PIR in which each user has access to more than one helper cache.

In this paper, we consider the problem of cache-aided multi-user PIR and propose a scheme with a slight variation in the placement phase. All the users are connected to databases via shared links. In the placement phase, a database will broadcast the content for the caches of all users (precisely defined in Section \ref{sec3}), and users will collude and decide which part of the broadcasted content will be stored in the cache of each user. That means the database will know the content placed in the caches of all users but will be unaware of which part of the whole content is stored in a user's cache. The cache contents of any two users are disjoint. This variation will help in maintaining the privacy constraint at a reduced rate.

\textit{Notations:} For an integer $N,$ $[N]$ denotes the set $\{1,2, \ldots, N\},$ and $A_{1:N}$ denotes the set $\{A_1, A_2, \ldots, A_N\}.$

\subsection{Cache-aided multi-user private information retrieval \cite{ZWSJC2021}}
In cache-aided MuPIR, there are $S$ non-colluding databases, and each database stores $N$ files $W_1, W_2, \ldots, W_N$, each of size $L$ bits. There are $K$ users connected to the databases via error-free links. Every user has a dedicated cache, which has the capability of storing $ML$ bits. Let $Z_k$ denotes the cache content of user $k$, $k \in [K]$. All caches are filled in the placement phase before knowing the demands of the users. In the delivery phase, all the users cooperatively generate query set $Q_s$ for each database $s \in [S]$. Using the answers $A_s, s \in [S]$ to the queries and the cache contents, users satisfy their demands without revealing any knowledge of the demand vector to the databases. Let $\theta = (d_1, d_2, \ldots, d_K)$ be the demand vector, where $d_i \in [N]$, for all $i \in [K]$. After receiving all the answers, the users should be able to recover their desired file correctly, i.e.,
$$H(W_{d_k} \ | \ Q_{1:S}, A_{1:S}, Z_k) = 0, \quad \forall k\in [K].$$
To preserve the privacy of demands with respect to the databases, the following condition should be satisfied
$$I(\theta; Q_s, A_s, W_{1:N}, Z_{1:K}) =0, \quad \forall s \in [S].$$
Let $D$ denotes the total number of bits transmitted by the databases to satisfy the demands of all the users. Then the rate (load) $R$ is defined as
$$R=\frac{D}{L}.$$

%

\begin{figure}[!t]
\centering
\includegraphics[width=3.5in]{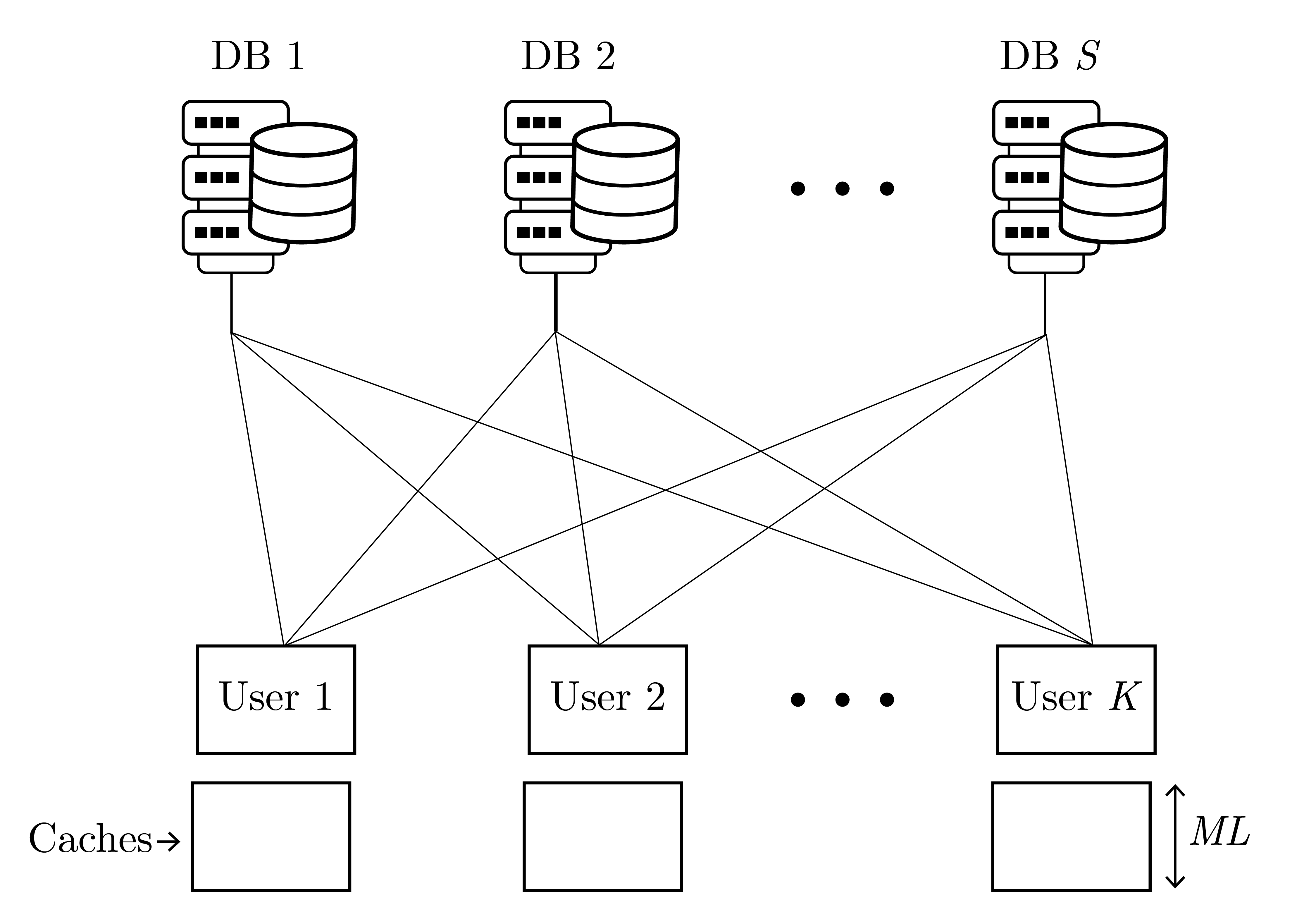}
\caption{Cache-aided multi-user private information retrieval}
\label{fig_1}
\end{figure}

\subsection{Our contributions }
The main contributions and outline of this paper are as follows.

\begin{itemize}
\item In Section \ref{sec2}, we provide an alternate scheme for PIR with subpacketization level $S^{N-1}$ and optimal rate $R=1+\frac{1}{S}+\frac{1}{S^2}+\cdots+\frac{1}{S^{N-1}}$. A scheme was given by Sun and Jafer \cite{SJ2017B} with the same subpacketization level and rate. The purpose of giving this alternate scheme is to provide a background for the scheme given in Section \ref{sec3} for cache-aided MuPIR, as the method of generating query sets in that scheme is inspired by this PIR scheme.

\item Using the PIR scheme given in Section \ref{sec2}, we present an idea to improve the product design (PD) scheme in terms of subpacketization level in Section \ref{sec3} and give an example to illustrate it.

\item In Subsection \ref{subsec3.1}, we propose a scheme for cache-aided MuPIR for a small value of cache size, which is inspired by the CFL (Chen, Fan, Letaief) coded caching scheme \cite{CFL2016} and the PIR scheme given in Section \ref{sec2}.

\item The subpacketization level in the proposed scheme is $KS^{N-1}$ which is less than the subpacketization level of the PD scheme which is ${K \choose t} S^N$, where $t=KM/N$.

\item In Section \ref{sec4}, we compare the rate of the proposed scheme $R$ and the rate of the PD scheme $R_{PD}$ for the value of $M$ given in \eqref{M}, and show that $R_{PD}-R >0$ for all $S \geq 2$ and $N \geq 2$. In Figure \ref{figc1}, the comparison of the rate of the PD scheme and proposed scheme (given only for one memory point) is given. 

\begin{figure}
\caption{Rate of the proposed scheme compared to the rate of the PD scheme\label{figc1}}
\centering
\includegraphics[scale=0.4]{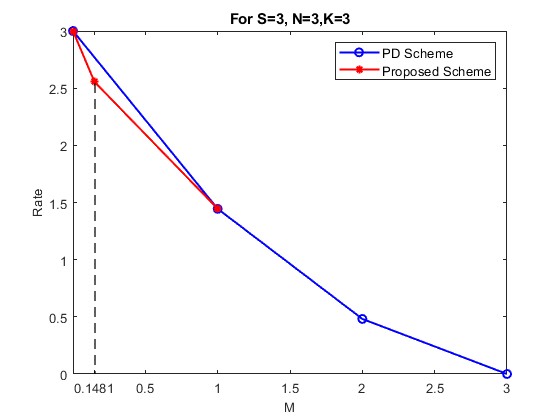}
\includegraphics[scale=0.4]{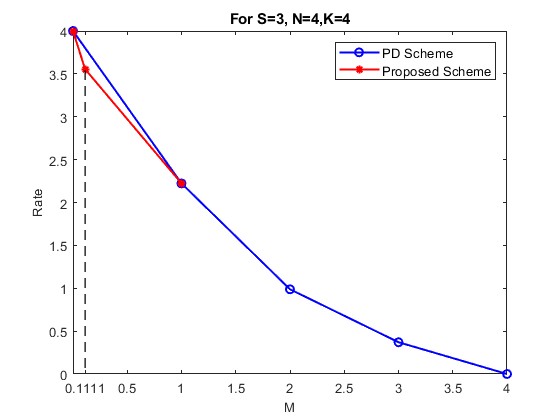}
\end{figure}


\end{itemize}

\section{PIR scheme with minimum subpacketization} \label{sec2}
Consider $S$ non-colluding databases containing $N$ files $W_1, W_2, \ldots, W_N$ each of size $L$ bits. The user is connected to the databases via an error-free link and wants a file $W_{\theta}$, for some $\theta \in [N]$. The user will generate query sets for each database in such a way that the individual database will not have any idea about the user's demand.
This problem was solved by Sun and Jafer \cite{SJ2017A} with subpacketization level $S^N$ and optimal rate $R=1+\frac{1}{S}+\frac{1}{S^2}+\cdots+\frac{1}{S^{N-1}}$. After that subpacketization level has been reduced to $S^{N-1}$ in \cite{SJ2017B} with the same rate. In this section, we provide an alternate method (Algorithm \ref{algo1} given in this section) to generate the query sets for the databases with the latter case, i.e., with the subpacketization level $S^{N-1}$. 
In this scheme, the query set of one database is different, and for the rest, it is the same.
We start with an example before describing the general scheme.
\begin{example}\label{ex2.1}
Consider a PIR problem with $S=4$ databases and $N=3$ files denoted by $A,B$ and $C$. Divide each file into $S^{N-1}=16$ subfiles. The files $A,B$ and $C$ are divided into the subfiles $\{A_1 A_2, \ldots, A_{16}\}$, $\{ B_1,B_2, \ldots, B_{16} \}$ and $\{C_1,C_2, \ldots, C_{16}\}$, respectively. Now the user generates three random permutations of $\{ 1, 2, \ldots, 16 \}$, namely, $P_A, P_B$ and $P_C$ for the subfiles of files $A, B$ and $C$, respectively. These permutations are not known to the databases. Let image of $A_i$ be represented by $a_i$ under the permutation $P_A$, i.e., $a_i=P_A(A_i)$ for all $i \in \{ 1, 2, \ldots, 16 \}$. Similarly, $b_j=P_B(B_j)$ and $c_k=P_C(C_k)$ for all $j, k \in \{ 1, 2, \ldots, 16 \}$. If the demand of the user is $A$ (without loss of generality), then the generated query sets are given in Table \ref{tb1}.

\begin{table}[!t]
\caption{Query table for Example \ref{ex2.1}\label{tb1}}
\centering
\begin{tabular}{ |c|c|c|c| } 
 \hline
 DB1 & DB2 & DB3 & DB4 \\
\hline 
$a_1$ & $a_2+b_1$ & $a_4+b_1$ & $a_6+b_1$ \\ 
$b_1$ & $a_3+c_1$ & $a_5+c_1$ & $a_7+c_1$ \\ 
$c_1$ & $b_2+c_2$ & $b_3+c_3$ & $b_4+c_4$ \\
$a_8+b_2+c_2$  & $a_{11}+b_3+c_3$ & $a_{13}+b_2+c_2$ & $a_{15}+b_2+c_2$ \\
$a_9+b_3+c_3$ & $a_{12} +b_4+c_4$ & $a_{14}+b_4+c_4$ & $a_{16}+b_3+c_3$ \\
$a_{10}+b_4+c_4$ & & & \\
 \hline
\end{tabular}
\end{table}

After receiving the answers to all queries from each database, the user will get all the subfiles of file $A$. Since $P_A, P_B$ and $P_C$ are all random permutations of $\{ 1, 2, \ldots, 16 \}$ and the query set of the individual database is symmetric along all the files, databases will not have any information regarding the demand of the user.
\end{example}

\subsection{PIR scheme} \label{subsec2.1}

Divide each file $W_i$ into $S^{N-1}$ subfiles denoted by $W_i(1), W_i(2), \ldots, W_i(S^{N-1})$, for all $i\in [N]$. Now we describe some terminologies which are used in this scheme.
\begin{itemize}

\item $k$-sum: The terminology $k$-sum is used to denote an expression representing the sum of $k$ distinct variables. For example $W_{i_1}\oplus W_{i_2} \oplus \cdots \oplus W_{i_k}$ is a $k$-sum which corresponds to the set $U=\{i_1, i_2, \ldots, i_k\} \subseteq [N]$.

\item Define a functiom $\phi_i(S,k)$ for all $i \in [S]$ and $k \in [N]$, as
\begin{equation} \label{eq1}
\phi_i(S,k)= \begin{cases}
g(S,k), \quad \text{for}\ i=1, \\
f(S,k), \quad \text{for} \ 2 \leq i \leq S
\end{cases},
\end{equation}
where $g(S,k)=\frac{S-1}{S} \left[ (-1)^{k-1} +(S-1)^{k-2} \right]$ and $f(S,k)=\frac{1}{S} \left[ (-1)^{k}+ (S-1)^{k-1} \right]$.

The function $\phi_i(S,k)$ will represent the number of repetitions of a $k$-sum in the query set of $i$-th database in this scheme. The structure of query set of one database is different from the other databases (without loss of generality, we have taken that the query structure of database $1$ is different from the other databases). Therefore, the value of the function $\phi_i(S,k)$ is different for $i=1$ and same for all $i \in \{2, 3, \ldots, S\}$.

\end{itemize}

\begin{algorithm} 
\caption{Generates query sets for each database $s \in S$ for PIR scheme [\ref{subsec2.1}].\label{algo1}}
\SetAlgoLined
\KwData{ Permutations $a^i$ for all $i \in [N]$, demand $d$}

\KwResult {Query sets $Q_{j},\  \forall j \in [S]$}

 $Q_1 \gets \{\{W_{i}(a^i_1) \} \ | \ i \in [N] \}$\; 

$Q_j \gets \emptyset, \  \forall j=2,3, \ldots, S$\;

$k\gets 2$\; $t_i \gets 1, \ \forall i \in [N]$\;

\While {$t_d < S^{N-1}$}
{
	\For {$j = 1 \to S$}
	{
		$I \gets 0$\;
		\For {$(i=1 \to S ) \ \& \ (i \neq j)$}
		{
			\For {$(Q \in Q_i) \ \& \ (|Q|=k-1)$}
			{
				\If {$Q \cap \{W_d(x) \ | \ x \in [S^{N-1}]\} = \emptyset$}
				{
					\begin{equation}\label{add1}
						\hspace{-4mm} Q_j \gets Q_j \cup \{Q \cup W_d(a^d_{++t_d})\};
					\end{equation}
					$I \gets 1$;
				}
			}
		}
		\If {$I=1$}
		{
			\For {all distinct $\{ i_1,i_2, \ldots , i_k \} \in [N]\backslash\{d\}$ }
			{
				\For {$t \in [\phi_j(S,k)]$}
				{
					\begin{equation}\label{add2}
					Q_j \gets Q_j \cup \{ \{ W_{i_1}(a^{i_1}_{++t_{i_1}}),  W_{i_2}(a^{i_2}_{++t_{i_2}}), \ldots,   W_{i_k}(a^{i_k}_{++t_{i_k}})\} \};
					\end{equation}
				}
			}
		}
	}
	$k \gets k+1$\;
}
\textbf{Note}: The notation $W_{i,j}(x)$ denotes the subsubfile $W_{i,j}^x$.

\end{algorithm} 

Suppose the user's demand is $W_d$, where $d \in [N]$. Let $\bm{Q}_j$ denote the query set for the database $j, j \in [S]$. For $N$ random permutations $a^i=(a^i_1, a^i_2, \ldots, a^i_{S^{N-1}})$ of set $[S^{N-1}]$ generated by the user, where $i \in [N]$ and demand $d$, run Algorithm \ref{algo1}. As an output of Algorithm \ref{algo1}, we will get $S$ sets $Q_j, j \in [S]$, and every set $Q_j$ corresponds to the query set $\bm{Q}_j$ of database $j$ in the following manner:
$$\bm{Q}_j = \left\{ \sum_{x \in Q} x \ \big| \ \forall Q \in Q_j \right\},$$
for all $j \in [S]$. For $j \in [S]$, adding a set $\{ W_{i_1}, W_{i_2}, \ldots, W_{i_k} \}$ in $Q_j$ is equivelent to adding a $k$-sum $W_{i_1} \oplus W_{i_2} \oplus  \cdots \oplus W_{i_k}$ in the query set $\bm{Q}_j$. In further discussion, we will also address $Q_j$ as query set and each set of size $k$ in $Q_j$ as $k$-sum.
The idea of this algorithm is mainly based on two steps: \\
\textbf{Step 1:} Use the extra information gained from one database in the query sets of other databases. \\
\textbf{Step 2:} After asking the queries in the form of $k$-sum, $k \in [N]$, which contains the desired subfiles, maintain the symmetry along all the files by adding the remaining possible $k$-sums. \\
The first query from database $1$ demands the first subfile (after applying the permutation) of each file $W_i, i \in [N]$. Since the demand of the user is $W_d$, all other subfiles $W_i(a^i_{1})$, where $i \in [N]$ and $i \neq d$ are the extra information. Now use this extra information for getting the next desired subfiles of $W_d$ from all other databases, i.e., the following queries will be added to the query set $\bm{Q}_j$ of database $j,  1 < j \leq S$,
$$
\{ W_d(a^d_{x})+ W_{i}(a^{i}_{1}) \ | \ x = 2+(j-2)(N-1), \ldots, 1+(j-1)(N-1); i \in [N], i \neq d \}. \nonumber
$$
Now the query set $Q_j$ contains only $(N-1)$ number of $2$-sums, all containing the subfiles of $W_d$, while the number of the total possible $2$-sums is $N \choose 2$. Therefore, to maintain symmetry across the files, add ${N\choose2}-(N-1)$ remaining $2$-sums using new subfiles of the files $W_i, i \neq d$ in the query set $Q_j$. Now the query sets of every database except the first database contain the extra information. Further, exploit this information gained from a database in the form of $2$-sums for getting subfiles of $W_d$ using $3$-sums in the query sets of other databases. Then again, maintain symmetry along all the files by adding the remaining possible $3$-sums. Also, the repetition of every $k$-sum should be equal while performing Step 2. Repeat this process until we find all the subfiles of $W_d$.

The total number of $k$-sums in $N$ files is $N\choose k$. Suppose every $k$ sum is repeated $\phi_j(S,k)$ times in the query set $\bm{Q}_j$. Since the structure of the query set of the first database is different from the other databases, the value of the function $\phi_i(S,k)$ is different for $i=1$ and the same for all $i \in \{2, 3, \ldots, S\}$, as defined in \eqref{eq1}. Now we explain the process of obtaining the values of the functions $g(S,k)$ and $f(S,k)$.

\begin{table}[!t]
\caption{The values of functions $g(S, k)$ and $f(S, k)$ in the given PIR scheme\label{tb2}}
\centering
\begin{tabular}{ |c||c|c| }
 \hline
$k$ & $g(S,k)$ & $f(S,k)$ \\
\hline
$1$ & $1$ & $0$ \\
$2$ & $0$ & $1$ \\
$3$ & $(S-1)f(S,2)$ & $(S-2)f(S,2)+g(S,2)$ \\
$4$ & $(S-1)f(S,3)$ & $(S-2)f(S,3)+g(S,3)$ \\
$\vdots$ & $\vdots$ & $\vdots$ \\
$t$& $(S-1)f(S,t-1)$ & $(S-2)f(S,t-1)+g(S,t-1)$ \\
$\vdots$ & $\vdots$ & $\vdots$ \\
$N$ & $(S-1)f(S,N-1)$ & $(S-2)f(S,N-1)+g(S,N-1)$ \\
\hline
\end{tabular}
\end{table}

\begin{itemize}
\item Since we are asking single subfiles ($1$-sum) only from database $1$, we have $g(S,1)=1$ and $f(S,1)=0$ .
\item The extra information gained from database $1$ in the first step can only be used in the query sets of other databases in the form of $2$-sums, and every $2$-sum will occur only once in a query set $\bm{Q}_j, j \in \{ 2,3, \ldots, S \}$. Therefore, $g(S,2)=0$ and $f(S,2)=1$.
\item In query set of database $1$, while using the extra information (in the form of $2$-sums) gained from databases $2,3, \ldots, S$, every $3$-sum is repeated $(S-1)$ times. In the query set of database $j, j \in \{2,3, \ldots, S\}$, every $3$-sum is repeated $(S-2)$ times.
\item Continuing in the same way, for $k \in [N]$, we have $g(S,k)=(S-1)f(S,k-1)$ because in the query set of database $1$ we will use extra information gained from databases $2,3, \ldots, S$ using $(k-1)$-sums which were already occurring $f(S,k-1)$ times. Similarly, we get, $f(S,k)=(S-2)f(S,k-1)+g(S,k-1)$.
\end{itemize}
By following the process given above, the values of functions $g(S,k)$ and $f(S,k)$ are given in Table \ref{tb2}.
By solving the recurrence relations $g(S,k)=(S-1)f(S,k-1)$ and $f(S,k)=(S-2)f(S,k-1)+g(S,k-1)$ with initial conditions $g(S,1)=1$, $f(S,1)=0$, $g(S,2)=0$ and $f(S,2)=1$, we get
$$g(S,k)=\frac{S-1}{S} \left[ (-1)^{k-1} +(S-1)^{k-2} \right],$$
$$f(S,k)=\frac{1}{S} \left[ (-1)^{k}+ (S-1)^{k-1} \right].$$

This rate of this scheme is $R=\left(1+\frac{1}{S}+\frac{1}{S^2}+ \cdots + \frac{1}{S^{N-1}} \right) ,$ which is optimal. The proof of correctness, privacy and achieving capacity for this scheme is given in Appendix \ref{asec1}.

\section{Cache-aided multi-user private information retrieval} \label{sec3}
Consider $S$ ($\geq 2$) non-colluding databases and each database stores $N$ ($\geq 2$) files $W_1, W_2, \ldots$, $W_N$ each of size $L$ bits. There are $K$ users connected to the databases via error-free links. Every user has a dedicated cache, which has the capability of storing $ML$ bits. Let $Z_u$ denotes the cache content of user $u$, $u \in [K]$. All caches are filled in the placement phase before knowing about the demands of the users. In the delivery phase, all the users cooperatively generate query sets for each database. Using the answers to the queries and the cache contents, users satisfy their demands without revealing any knowledge of the demand vector to the databases.

In this section, first we present an idea to improve the PD scheme given in \cite{ZWSJC2021} in terms of  subpacketization level and give an example. Then we present our scheme which is inspired by the CFL coded caching scheme \cite{CFL2016} and the PIR scheme given in Section \ref{sec2}.

In the PD scheme, authors used the combination of MAN coded caching scheme \cite{MN2014} and SJ (Sun and Jafar) PIR scheme \cite{SJ2017A}.  The subpacketization level in that scheme was ${K \choose t} S^{N}$, where $t=\frac{KM}{N}$. If we use the PIR scheme given in Section \ref{sec2} instead of using the SJ scheme in the PD scheme, then we get the same rate but with reduced subpacketization level of ${K \choose t} S^{N-1}$, where $t=\frac{KM}{N}$.
\begin{example}\label{ex0}
Consider the cache-aided MuPIR problem with $S=2, N=3, K=2$ and $M=3/2$. \\
\textit{Placement}: Divide each file $A,B,C$ into two subfiles, i.e., $A=(A_1, A_2)$, $B=(B_1, B_2)$ and $C=(C_1, C_2)$. The cache content of user $k \in \{ 1, 2 \}$ is $Z_k=\{ A_k, B_k, C_k \}$. \\
\textit{Delivery}: In the PD scheme, each subfile is divided into $8$ subsubfiles. Let $[A_i^1, A_i^2,$ $\ldots, A_i^{8}]$, $[B_i^1, B_i^2, \ldots, B_i^{8}]$ and $[C_i^1, C_i^2, \ldots, C_i^{8}]$ represent random permutations of $8$ subsubfiles of $A_i, B_i$ and $C_i$, respectively, generated by the users, where $i \in \{ 1,2\}$. Then the queries are given in Table \ref{ex0:tb1}(a) for the demand vector $(1,2)$.

\begin{table}[!t]
\caption{Query tables for Example \ref{ex0} \label{ex0:tb1}}
\centering
\resizebox{\columnwidth}{!}{%
\begin{tabular}{ cc }   
\begin{tabular}{ |c|c| }
 \hline
DB1 & DB2 \\
\hline
$A_2^1 +A_1^1 $            &            $A_2^2 +A_1^2 $   \\
$B_2^1 + B_1^1$		&		  $B_2^2 + B_1^2$ \\
$C_2^1 + C_1^1$		&		  $C_2^2 +C_1^2$ \\
$A_2^3+ B_2^2 + B_1^3+A_1^2$		&		  $A_2^5 + B_2^1 + B_1^5 + A_1^1$ \\
$A_2^4+ C_2^2 + C_1^3+A_1^3$		&		  $A_2^6 + C_2^1 + C_1^4 + A_1^4$ \\
$B_2^3+ C_2^3 + C_1^2+B_1^4$		&		  $B_2^4 + C_2^4 + C_1^1 + B_1^6$ \\
$A_2^7+ B_2^4 + C_2^4 + C_1^4+  B_1^7+A_1^4$		&		  $A_2^8+ B_2^3+ C_2^3 + C_1^3+  B_1^8+A_1^3$	 \\
\hline
\end{tabular} &  
\begin{tabular}{ |c|c| }
 \hline
DB1 & DB2 \\
\hline

$A_2^1 +A_1^1 $            &            $A_2^2 +B_2^1 + B_1^2 + A_1^1 $   \\
$B_2^1 + B_1^1$		&		  $A_2^3 + C_2^1 + C_1^2 + A_1^2$ \\
$C_2^1 + C_1^1$		&		  $B_2^2 +C_2^2  + C_1^1 + B_1^3$ \\
$A_2^4 + B_2^2 + C_2^2 + C_1^2 +B_1^4 + A_1^2$ &  \\
\hline
\end{tabular} \\
(a) & (b) \\
\end{tabular}
}
\end{table}


If we use the PIR scheme given in Section \ref{sec2} instead of using the SJ scheme, then each subfile is divided into $4$ subsubfiles. Let $[A_i^1, A_i^2, A_i^3, A_i^{4}]$, $[B_i^1, B_i^2, B_i^3, B_i^{4}]$ and $[C_i^1, C_i^2, C_i^3, C_i^{4}]$ represent random permutations of $4$ subsubfiles of $A_i, B_i$ and $C_i$, respectively, generated by the users, where $i \in \{ 1,2 \}$. Then the queries are given in Table \ref{ex0:tb1}(b) for the demand vector $(1,2)$.
  
%

\end{example}


\subsection{CFL coded caching scheme \cite{CFL2016}}\label{subsec3.0}
For a coded caching problem with $N$ files and $K$ users, CFL scheme gives the optimal rate for the cache size $M=\frac{1}{K}$ when $N\leq K$ . In orignal version of CFL scheme given in \cite{CFL2016}, each file is divided into $NK$ subfiles but it is sufficient to divide each file in $K$ subfiles. Therefore divide each file $W_i$ in $K$ subfiles denoted as $W_{i,1}, W_{i,2}, \ldots, W_{i,K}$ for all $i \in [N]$. 

\noindent\textbf{Placement phase:} The cache content of user $k \in [K]$ is
$$Z_k=\{ W_{1,k} \oplus W_{2,k} \oplus \cdots \oplus W_{N,k} \}.$$

\noindent\textbf{Delivery phase:} Let the demand vector be $d=(d_1,d_2, \ldots, d_K)$ and each file is demanded by at least one user. Without loss of generality assume that first $N$ users demanded all $N$ files, i.e., $(d_1,d_2, \ldots,d_N)=[N]$.
Then transmissions are as follows.
\begin{itemize}
\item For all $j\in [K]$, transmit subfiles $W_{1,j}, W_{2,j}, \ldots, W_{d_j-1,j}, W_{d_j+1,j}, \ldots, W_{N,j}$.
\item For all $j \in \{N+1, N+2, \ldots, K\}$, transmit $W_{d_j, j} \oplus W_{d_{j'},j'}$, where $j' \in [N]$ and $d_j=d_{j'}$.
\end{itemize}
The rate of this scheme is $R=N\left(1-\frac{1}{K} \right)$, which is optimal as it satisfy the cut-set lower bound \cite{MN2014}.

\subsection{Our proposed scheme} \label{subsec3.1}
In this subsection, we propose a scheme for cache-aided MuPIR for a specific value of cache size, and this scheme is inspired by the CFL coded caching scheme and the PIR scheme given in Section \ref{sec2}. Since the CFL scheme achieves the lower peak rate of caching if the number of users is no less than the number of files in the database, our scheme gives a better rate than the PD scheme for $N \leq K$. 


We are adopting the notation of $k$-sum and the function $\phi_{i}(S,k)$ from Section \ref{sec2} and defining some new notations which are used in this scheme.

\begin{itemize}

\item Define a function $\psi_i(S,k)$ for all $i \in [S]$ and $k \in [N]$, which denotes the number of repetition of a $k$-sum in the query to $i$-th database in this scheme.
$$\psi_i(S,k)=\left \lceil \frac{k(N-1)}{N} \ \phi_i(S,k) \right \rceil .$$

\item Let $q$ be an integer defined as
\begin{equation}\label{q}
q=\sum_{k=1}^{N} \sum_{s=1}^{S} {N \choose k} \psi_s(S,k).
\end{equation}

\item In this scheme, cache size is defined as,
\begin{equation}\label{M}
M=\frac{NS^{N-1}-q}{KS^{N-1}}.
\end{equation}

\item For $i \in [N]$ and $s \in [S]$, define a function
$$\mathcal{F}_s^{i,d} (k)= \begin{cases}
{N-1 \choose k-1} \phi_s(S,k) & \text{if} \ i\neq d \\
{N \choose k} \psi_s(S,k) - (N-1) {N-1 \choose k-1} \phi_s(S,k)     & \text{if} \ i=d 
\end{cases},
$$
where $d \in [N]$.

\end{itemize}

\vspace{2mm}

\noindent Divide each file $W_i, i \in [N]$ into $K$ subfiles denoted by $W_{i,1}, W_{i,2}, \ldots, W_{i,K}$. Further, divide every subfile $W_{i,j}$ into $S^{N-1}$ subsubfiles denoted by $W_{i,j}^1, W_{i,j}^2, \ldots, W_{i,j}^{S^{N-1}}$ for all $i \in [N]$ and $j \in [K]$ .

\textbf{Subpacketization:} Since every file $W_i$ is divided into $KS^{N-1}$ subsubfiles, $\forall i \in [N]$, subpacketization level is $KS^{N-1}$. Hence the size of each file should be greater than or equal to $KS^{N-1}$, i.e., $L \geq KS^{N-1}$. 

\vspace{2mm}

\textbf{Placement phase:} In this phase, all the users fill\ their caches. Let the cache content of user $u$ be denoted by $Z_u$, $u \in [K]$. In the time when the network is not busy, a database will broadcast $Z$, where
$$
Z=\left\{ W_{1,j}^t \oplus W_{2,j}^t \oplus \cdots \oplus W_{N,j}^t \ | \ j \in [K],  t=H+1, H+2, \ldots, S^{N-1} \right\}. 
$$
and $H=q-(N-1)S^{N-1}$. Then users will cooperatively generate a random permutation of $[K]$, say it is $P=(p_1,p_2, \ldots , p_K)$. By using the broadcasted data from the database, user $u$ will fill its cache with cache content
$$
Z_u = \left\{ W_{1,p_u}^t \oplus W_{2,p_u}^t \oplus \cdots \oplus W_{N,p_u}^t \ | \ t=H+1, H+2, \ldots, S^{N-1} \right\}, $$
for all $u \in [K]$.
Since the number of subsubfiles stored in the cache of a user is $S^{N-1}-H$ and the size of one subsubfile is $\frac{L}{KS^{N-1}}$, the total size of cache content is
$$
\frac{(S^{N-1}-H) L}{KS^{N-1}}=\left( \frac{S^{N-1}-q+(N-1)S^{N-1}}{KS^{N-1}} \right) L  =\left( \frac{NS^{N-1}-q}{KS^{N-1}} \right) L= ML. 
$$


For the delivery phase, users will cooperatively generate the query sets for each database using Algorithm \ref{algo2} and Algorithm \ref{algo3}, and these algorithms use Function 1 and Function 2. So first, we describe the details of these functions.

{
\begin{algorithm}
\NoCaptionOfAlgo
\caption*{\textbf{Function 1:} Generates query sets for each database $s \in S$ for $j \in [K]$.\label{func1}}
  \SetKwFunction{FMain}{QSet1}
  \SetKwProg{Fn}{Function}{:}{}
  \Fn{\FMain{$j$; $a^i, \ \forall i \in [N]$; $d$}}{
 $Q'_1 \gets \{\{W_{i,j}(a^i_1) \} \ | \ i \in [N] \} ;$ 

$Q'_s \gets \emptyset, \  \forall s=2,3, \ldots, S ;$

$t_i \gets 1, \ \forall i \in [N] ;$

\For {$k=2 \to N$}
{ 
   $T_i \gets t_i, \ \forall i \in [N]$\;
   \For {$s=1 \to S$}
   {
	$\mathcal{C}\gets [U \subseteq [N] \ | \ |U|=k] (\psi_s(S,k)\ \text{times})$\;
	$\mathcal{T} \gets \emptyset$\;
	 \For {$i=1 \to N$}
	{
		  \For {$i'=1 \to \mathcal{F}_s^{i,d} (k)$}
		{
			\eIf {$\exists \ U \in \mathcal{C} \ \text{such that} \ i \in U$}
			{
				\begin{equation}\label{add2.1}
					Q'_s \gets Q'_s \cup \{ \{ W_{i,j}(a^i_{++t_i})  \}  \bigcup_{u \in U \backslash \{i\}} \{ W_{u,j} 
                                    (a^u_{\leq T_u}) \} \} ;
				\end{equation}
				$\mathcal{C} \gets \mathcal{C} - \{U\}$\;
				$\mathcal{T} \gets \mathcal{T} \cup \{U\}$\;
			}			
			{
				choose ($U \in \mathcal{C}$ \ $\&$ \ $V \in \mathcal{T},$ whose corresponding query is   
                           $q'=\{ W_{v,j}(a^v_{t'_v}) \ | \ v \in V \} \in Q'_s$) s.t. $(i \in V, v_1 \in U \cup V \ \& \ t'_{v_1} >T_{v_1})$ 
				\begin{align}\label{add2.2}
				Q'_s &\gets Q'_s \cup \{\{W_{v_1,j}(a^{v_1}_{t'_{v_1}}) \} 
                                 \bigcup_{v \in V \backslash \{v_1\}} \{ W_{v,j}(a^v_{\leq T_{v}}) \} ; \\ 
				q' &\gets q' - \{ W_{i,j}(a^i_{t'_i}), W_{v_1,j}(a^{v_1}_{t'_{v_1}}) \} 
				           \bigcup \{ W_{i,j}(a^i_{++t_i}), W_{v_1,j}(a^{v_1}_{\leq T_{v_1}}) \} ;\nonumber
				\end{align}
				$\mathcal{C} \gets \mathcal{C} - \{U\}$\;
				$\mathcal{T} \gets \mathcal{T} \cup \{U\}$\;
			}
		}
	}
    }
}
        \KwRet{$Q'_{s},\  \forall s \in [S]$}\; 
}		
\textbf{Note}: \textbf{1.} The notation $W_{i,j}(x)$ denotes the subsubfile $W_{i,j}^x$.
\textbf{2.} In expression $W_{u,j}(a^u_{\leq T_u})$, choose an integer $\leq T_u$ in such a way that maintains the symmetry in the repetition of subsubfiles of every subfile $W_{i,j}$ in query sets $Q'_s$.
\end{algorithm}
%

%

\noindent \textbf{Correctness, privacy and output of Function 1:}

Function 1 generates the query sets corresponding to the subfiles $W_{1,j}, W_{2,j}, \ldots, W_{N,j},$ where $j \in [K]$. Therefore, the inputs for this algorithm are the value of $j$, $N$ permutations of $[S^{N-1}]$ and an integer $d \in [N]$. The output of this function will produce $S$ query sets $Q'_s, s \in [S]$, one for each database. This function is a generalization of Algorithm \ref{algo1}. In Algorithm \ref{algo1}, the aim is to find all the subfiles of the desired file and maintaining the symmetry along all $N$ files in the query set for each database. However, in Function 1, the aim is to find all the subsubfiles of subfiles $W_{i,j}$, where $i \neq d$ and maintaining the symmetry along all $N$ subfiles in the query set for each database. While maintaining the symmetry along all $N$ subfiles, we also obtained the $H$ subsubfiles of subfile $W_{d,j}$. So the idea is to use the cache content of the user $j$ to get the remaining $S^{N-1}-H$ subsubfiles of $W_{d,j}$.
\begin{itemize}
\item \textbf{Correctness:} To prove the correctness of Function 1, we need to prove the following points.

\begin{enumerate}

\item The number of sets in collection $\mathcal{C}$ is equal to the number of sets used to make queries in $Q'_s$, where $s\in [S]$, i.e.,
$${N\choose k} \psi_s(S,k)=\sum_{i=1}^{N} \mathcal{F}_s^{i,d}(k).$$
\begin{proof}
We have
\begin{align*}
\sum_{i=1}^{N} \mathcal{F}_s^{i,d}(k) &= \mathcal{F}_s^{d,d}(k) + \sum_{i=1, i \neq d}^{N} \mathcal{F}_s^{i,d}(k) \\
&= {N\choose k} \psi_{s}(S,k)-(N-1) {N-1 \choose {k-1}} \phi_s(S,k) +   \sum_{i=1, i \neq d}^{N} {N-1 \choose {k-1}} \phi_s(S,k) \\
&= {N\choose k} \psi_{s}(S,k).
\end{align*}
\end{proof}

\item For each subfile $W_{i,j}, i \in [N]$, the total $\mathcal{F}_s^{i,d}(k)$ number of $k$-sums, containing a new subsubfile of $W_{i,j}$, are getting added to $Q'_s$. Hence there should be enough sets of size $k$ in collection $\mathcal{C}$ which contain $i$, i.e.,
$$\mathcal{F}_s^{i,d}(k) \leq {N-1 \choose{k-1}} \psi_s(S,k).$$

\begin{proof}
For $i \neq d$, we have
$
\mathcal{F}_s^{i,d}(k) = {N-1 \choose{k-1}} \phi_s(S,k)  \leq {N-1 \choose{k-1}} \psi_s(S,k),
$
as $ \phi_s(S,k) \leq  \psi_s(S,k)$. 
For $i=d$, we have
$\mathcal{F}_s^{d,d}(k) = {N \choose{k}} \psi_s(S,k) - (N-1)  {N-1 \choose{k-1}} \phi_s(S,k).$
Since we have 
$$
\psi_s(S,k)= \left \lceil \frac{k(N-1)}{N} \phi_s(S,k) \right \rceil \leq \left \lceil \frac{kN}{N} \phi_s(S,k) \right \rceil = k \phi_s(S,k), $$
we get
\begin{align*}
\mathcal{F}_s^{i,d}(k) &\leq {N \choose{k}} k\phi_s(S,k) - (N-1)  {N-1 \choose{k-1}} \phi_s(S,k) \\
&= {N-1 \choose{k-1}} \phi_s(S,k) \leq {N-1 \choose{k-1}} \psi_s(S,k).
\end{align*}
\end{proof}
\end{enumerate}

\item \textbf{Privacy:} In Function 1, queries in $Q'_s, s\in [S]$ are being generated corresponding to every set in collection $\mathcal{C}$, and every subset of $[N]$ of size $k$ is contained in $\mathcal{C}$ equal number of times (which is $\psi_s(S,k)$). Therefore, symmetry is maintained in $Q'_s$ along all the subfiles $W_{i,j}, i \in[N]$.

\item \textbf{Output:} The $k$-sums containing new subsubfiles of subfiles $W_{i,j},$ $i \in [N]$ are being added to $Q'_s,$  $s \in [S]$ only in \eqref{add2.1} and \eqref{add2.2}, and in both the equations only one subsubfile (added with increment function $++t$) is new, and others already appeared in the query sets of other databases. Since this function starts with adding $1$-sums in $Q'_1$, each subsubfile that appeared in any of the query sets can be decoded. Therefore, to find the number of all subsubfiles which can be obtained from the answers to the query sets generated by Function 1, we only need to find the number of subsubfiles that appeared in all the query sets.
\begin{itemize}
\item The number of subsubfiles of a subfile $W_{i,j}$, where $i \in [N]$ and $ i\neq d$, obtained by the answers to the query sets $Q'_s, s\in [S]$ is
$
\eta_i = 1+\sum_{k=2}^{N} \sum_{s=1}^S \mathcal{F}_s^{i,d}(k).
$
Since for $i \neq d$,
$
 \mathcal{F}_s^{i,d}(1) = \sum_{s=1}^S {N-1 \choose{0}} \phi_s(S,1) =g(s,1) +(S-1)f(S,1) = 1+0=1, 
$
 we have
$$\eta_i =\sum_{k=1}^{N} \sum_{s=1}^S \mathcal{F}_s^{i,d} (k) 
= \sum_{k=1}^{N} \sum_{s=1}^S {N-1 \choose{k-1}} \phi_s(S,k) 
= S^{N-1}.$$


\item The number of subsubfiles of subfile $W_{d,j}$ obtained by the answers to the query sets $Q'_s, s\in [S]$,
\begin{align*}
\eta_d &= 1+\sum_{k=2}^{N} \sum_{s=1}^S \mathcal{F}_s^{d,d} (k) = \sum_{k=1}^{N} \sum_{s=1}^S \mathcal{F}_s^{d,d} (k) \\
&= \sum_{k=1}^{N} \sum_{s=1}^S \left( {N \choose k} \psi_s(S,k) - (N-1) {N-1 \choose{k-1}} \phi_s(S,k) \right)\\
&= \sum_{k=1}^{N} \sum_{s=1}^S {N \choose k} \psi_s(S,k) - (N-1)  \sum_{k=1}^{N} \sum_{s=1}^S {N-1 \choose{k-1}} \phi_s(S,k) \\
&=q-(N-1)S^{N-1} = H.
\end{align*}
\end{itemize}

 Hence using the answers to the query sets $Q'_s, s\in [S]$ generated by Function 1, the users can decode all $S^{N-1}$ subsubfiles of a subfile $W_{i,j}$, where $i \in [N]$ and $i \neq d$, and $H$ subsubfiles of the subfile $W_{d,j}$.
\end{itemize}

\begin{algorithm}
\NoCaptionOfAlgo
\caption*{\textbf{Function 2:} Generates query sets for each database $s \in S$ for $j \in [K] \backslash \mathcal{B}$.\label{func2}}

\SetAlgoLined


 \SetKwFunction{FMain}{QSet2}
  \SetKwProg{Fn}{Function}{:}{}
  \Fn{\FMain{$\omega_i(x), i \in [N]$ ; $a^i, \ \forall i \in [N]$}}{

 $Q'_1 \gets \{\{\omega_i(a^i_1) \} \ | \ i \in [N] \} ;$ 

$Q'_s \gets \emptyset, \  \forall s=2,3, \ldots, S ;$

$t_i \gets 1, \ \forall i \in [N] ;$

\For {$k=2 \to N$}
{ 
   $T_i \gets t_i, \ \forall i \in [N]$\;
   \For {$s=1 \to S$}
   {
	$\mathcal{C}\gets [U \subseteq [N] \ | \ |U|=k] (k\phi_s(S,k)\ \text{times})$\;
	$\mathcal{T} \gets \emptyset$\;
	 \For {$i=1 \to N$}
	{
		  \For {$i'=1 \to {N-1 \choose {k-1}} \phi_s(S,k)$}
		{
			\eIf {$\exists \ U \in \mathcal{C} \ \text{such that} \ i \in U$}
			{
				\begin{equation}\label{add3.1}
					Q'_s \gets Q'_s \cup \{ \{ \omega_i(a^i_{++t_i})  \}  \bigcup_{u \in U \backslash \{i\}} \{ \omega_u(a^u_{\leq T_u}) \} \};
				\end{equation}
				$\mathcal{C} \gets \mathcal{C} - \{U\}$\;
				$\mathcal{T} \gets \mathcal{T} \cup \{U\}$\;
			}			
			{
				choose $(U \in \mathcal{C} \ \& \ V \in \mathcal{T}, \ \text{whose corresponding query is} \ q'=\{ \omega_v(a^v_{t'_v}) \ | \ v \in V \} \in Q'_s) \ \text{s.t.} \ (i\in V, v_1   
                                          \in U \cup V \ \& \ t'_{v_1} >T_{v_1}$) 
				\begin{align}\label{add3.2}
				Q'_s &\gets Q'_s \cup \{\{\omega_{v_1}(a^{v_1}_{t'_{v_1}}) \}  \bigcup_{v \in V \backslash \{v_1\}} \{ \omega_v(a^v_{\leq T_{v}}) \} ; \\
				 q' &\gets q' - \{ \omega_i(a^i_{t'_i}), \omega_{v_1}(a^{v_1}_{t'_{v_1}}) \} \ \bigcup \{ \omega_i(a^i_{++t_i}), \omega_{v_1}(a^{v_1}_{\leq T_{v_1}}) \} ; \nonumber
				\end{align}
				$\mathcal{C} \gets \mathcal{C} - \{U\}$\;
				$\mathcal{T} \gets \mathcal{T} \cup \{U\}$\;
			}
		}
	}
    }
}
 \KwRet{$Q'_{s},\  \forall s \in [S]$}\; 
}
 \textbf{Note}: \textbf{1.} The notation $W_{i,j}(x)$ denotes the subsubfile $W_{i,j}^x$.
\textbf{2.}  In expression $\omega_{u}(a^u_{\leq T_u})$, choose an integer $\leq T_u$ in such a way that maintains the symmetry in the repetition of subsubfiles of every subfile $\omega_{i}$ in query sets $Q'_s$.
\end{algorithm}

\noindent \textbf{Correctness, privacy and output of Function 2}

The input of Function 2 contains $N$ functions $\omega_i(x), i\in [N]$ of the subfiles $W_{i,j}, i \in[N], j \in [K]$, for example, $\omega_i(x)=W_{i,j_1}(x) + W_{i,j_2}(x)$ for some $i \in [N], j_1,j_2 \in[K]$ and $x\in[S^{N-1}]$. This function generates query sets for each database which are symmetric among all $N$ functions $\omega_i(x), i\in [N]$. Using the answers to each query set, the user will get the values of all $S^{N-1}$ function of subsubfiles corresponding to each $\omega_i, i\in [N]$, i.e., users will be able to find the values of $\omega_i(x)$ for all $i \in [N]$ and $x \in [S^{N-1}]$.
\begin{itemize}
\item \textbf{Correctness:} To prove the correctness of Function 2, we need to prove the following points.

\begin{enumerate}

\item The number of sets in collection $\mathcal{C}$ is equal to the number of sets used to make queries in $Q'_s$, where $s\in [S]$, i.e.,
$${N\choose k} \left( k \phi_s(S,k) \right)= N {N-1 \choose{k-1}} \phi_s(S,k),$$
which can be easily proved.

\item For each  $\omega_{i}, i \in [N]$, the total ${N-1 \choose{k-1}} \phi_s(S,k)$ number of $k$-sums, containing a new subsubfiles corresponding to $\omega_{i}$, are getting added to $Q'_s$. Hence there should be enough sets of size $k$ in collection $\mathcal{C}$ which contain $i$, i.e.,
$${N-1 \choose{k-1}} \phi_s(S,k) \leq {N-1 \choose{k-1}} (k \phi_s(S,k)),$$
which is true.

\end{enumerate}

\item \textbf{Privacy:} In Function 2, queries in $Q'_s,  s\in [S]$ are being generated corresponding to every set in collection $\mathcal{C}$, and every subset of $[N]$ of size $k$ is contained in $\mathcal{C}$ equal number of times (which is $k\phi_s(S,k)$). Therefore, symmetry is maintained in $Q'_s$ along all the functions $\omega_{i}, i \in[N]$.

\item \textbf{Output:} The $k$-sums contaning new subsubfiles of $\omega_{i}. i \in [N]$ are being added to $Q'_s, s \in [S]$ only in \eqref{add3.1} and \eqref{add3.2}, and in both the equations only one subsubfile (added with increment function $++t$) is new, and others already appear in the query sets of other databases. Since this function starts with adding $1$-sums in $Q'_1$, each subsubfile that appeared in any of the query sets can be decoded. Therefore, to find the number of all subsubfiles which can be obtained from the answers to the query sets generated by Function 2, we only need to find the number of subsubfiles of $\omega_i, i \in [N]$ that appeared in all the query sets which is equal to
$$\eta_i = 1+\sum_{k=2}^N \sum_{s=1}^S {N-1 \choose{k-1}} \phi_s(S,k) = \sum_{k=1}^N \sum_{s=1}^S {N-1 \choose{k-1}} \phi_s(S,k) =S^{N-1}.$$
Hence using the answers to the query sets $Q'_s, s\in [S]$ generated by Function 2, the users can decode all $S^{N-1}$ subsubfiles of $\omega_{i}$, where $i \in [N]$.

\end{itemize}

\textbf{Delivery phase:}
Suppose the demand vector for all $K$ users is $\theta=(d_1, d_2, \ldots, d_K)$. All the users will cooperatively generate the query sets for each database using Algorithm \ref{algo2} and Algorithm \ref{algo3}. The generation of the query sets is different for the cases $N=K$ and $N<K$.

\subsubsection{For $N=K$}
%

\begin{algorithm}
\caption{Generates query sets for each database $s \in S$ for the proposed scheme for $N=K$. \label{algo2}}

\SetAlgoLined
\KwData{ Permutations $P^{\lambda i}, \forall i \in [N], \lambda \in [K]$, demand vector $\theta$}

\KwResult {Query sets $Q_{s},\  \forall s \in [S]$}
\For {$\lambda=1 \to K$}
{
		QSet1 ($p_{\lambda}; P^{\lambda i}, \ \forall i \in [N] ; d_{\lambda}$);

		$Q^{\lambda}_s \gets \left\{ \sum_{x \in Q} x \ | \ Q \in Q'_s \right\}, \ \forall  \ s \in[S] ;$
}
		$Q_s \gets \bigcup_{\lambda=1}^K Q^{\lambda}_s, \ \forall  \ s \in[S] ;$

\end{algorithm}

Here, we only consider the case when all $K$ users have distinct demands with the demand vector $\theta=(d_1, d_2, \ldots, d_K)$. For each $c \in [K]$, the users generate $N$ permutations $\{P^{c1},P^{c2}, \ldots, P^{cN}\}$ of the set $[S^{N-1}]$ such that
$P^{ci} = (p^i_1, p^i_2, \ldots, p^i_{S^{N-1}} ) \ \text{for all} \ i \in[N] $ and $(p^{d_{c}}_{H+1}, p^{d_{c}}_{H+2}, \ldots, p^{d_{c}}_{S^{N-1}}) = (H+1, H+2, \ldots, S^{N-1}).$
In Algorithm \ref{algo2}, for each $c \in [K]$, Function 1 (QSet1) is called with the following values of the parameters
$$p_c  \longrightarrow j; \ \ \ P^{ci} \longrightarrow a^i, \forall i \in[N]; \ \ \ d_{c} \longrightarrow d ,$$

and, we get query set $Q_s$ for each database $s \in [S]$ as the output.


\subsubsection*{Decoding}
Now we prove that after receiving answers to all the $Q_s, s\in [S]$ from each database, all the users will get their desired files.
Consider a user $c' \in [K]$ with demand of file $W_{d_{c'}}$. The file $W_{d_{c'}}$ has $K$ subfiles $W_{d_{c'},1}, W_{d_{c'},2}, \ldots, W_{d_{c'},K}$.
All the subsubfiles of subfile $W_{d_{c'},j}$, where $j \in [K]$ and $j \neq p_{c'}$ are obtained  when Algorithm \ref{algo2} calls Function 1 for $j \in [K]\backslash \{{c'}\}$, as all the users have distinct demands. Now, only $W_{d_{c'}, p_{c'}}$ subfile is left to be obtained for user $c'$.
First $H$ subsubfiles of $W_{d_{c'}, p_{c'}}$ are obtained  when Algorithm 1 calls Function 1 } for $c'$, as 
$(p^{d_{c'}}_{H+1}, p^{d_{{c'}}}_{H+2}, \ldots, p^{d_{{c'}}}_{S^{N-1}}) = (H+1, H+2, \ldots, S^{N-1}).$
 Now remaining $(S^{N-1}-H)$ subsubfiles of $W_{d_{c'}, p_{c'}}$ will be obtained using cache content of user $c'$,
$
Z_{c'} = \left\{ W_{1,p_{c'}}^t \oplus W_{2,p_{c'}}^t \oplus \cdots \oplus W_{N,p_{c'}}^t \ |  \ t=H+1, H+2, \ldots, S^{N-1} \right\},
$
as we already have all the subsubfiles of $W_{i, p_{c'}}$, where $i\in [K]$ and $ i \neq d_{c'}$, by calling Function 1 for $c'$ in Algorithm \ref{algo2}.

\begin{theorem}\label{thm4.1}
The proposed scheme achieves the rate $\frac{q}{S^{N-1}}$ for $N=K$.
\end{theorem}
\begin{proof}
Since in Algorithm \ref{algo2}, Function 1 is called for all $c\in [K]$, the total number of queries is $Kq'$, where $q'$ is the total number of queries generated by Function 1, which can be computed as follows.
\begin{align*}
q'&= \sum_{s=1}^S |Q'_s|= N+\sum_{s=1}^S \sum_{k=2}^N {N \choose k} \psi_{s}(S,k) \\
&=\sum_{s=1}^S {N \choose 1} \psi_{s}(S,1) + \sum_{s=1}^S \sum_{k=2}^N {N \choose k} \psi_{s}(S,k) \left( \text{as} \ \psi_s(S,1) =\begin{cases} 1 & \ \text{if}\ s=1 \\ 0 & \ \text{if}\ s\neq 1 \end{cases} \right) \\
&=\sum_{s=1}^S \sum_{k=1}^N {N \choose k} \psi_{s}(S,k) = q.
\end{align*}

Therefore, the total number of queries is $Kq$. Since the answer to every query of size of a subsubfile, Rate
$R=\frac{Kq}{KS^{N-1}} = \frac{q}{S^{N-1}}.$
 \end{proof}

\begin{example}\label{ex4.1}
Consider a PIR problem with $S=3$ servers, $N=3$ files and $K=3$ users. Each file is divided into $3$ subfiles, i.e.,
$W_i  \rightarrow W_{i,1}, W_{i,2}, W_{i,3} \ \ \forall i \in \{1, 2, 3\}.$
Each subfile is further divided into $9$ subsubfiles, i.e.,
$W_{i,j} \rightarrow W_{i,j}^1, W_{i,j}^2, \ldots, W_{i,j}^9, \ \forall i,j \in \{1, 2, 3\}.$
Now we compute the values of function $\psi_s(S,k)$, where $s, k \in \{1, 2, 3\}$ in Table \ref{tb3}.
\begin{table}[!t]
\caption{The values of functions $\phi_s(S,k)$ and $\psi_s(S,k)$ in  Example \ref{ex4.1}\label{tb3}}
\centering
\begin{tabular}{ cc }   
\begin{tabular}{|c|c|}
 \hline
$g(3,1)=1$ & $f(3,1)=0$ \\
$g(3,2)=0$ & $f(3,2)=1$ \\
$g(3,3)=2$ & $f(3,3)=1$ \\
\hline
\end{tabular} &  
\begin{tabular}{|c|c|}
 \hline
$G(3,1)=1$ & $F(3,1)=0$ \\
$G(3,2)=0$ & $F(3,2)=2$ \\
$G(3,3)=4$ & $F(3,3)=2$ \\
\hline
\end{tabular} \\
\end{tabular}
\end{table}
We have
$q=\sum_{k=1}^3 {3 \choose k} (G(3,k)+2F(3,k))=23$
and $H=q-(N-1)S^{N-1}=5$.
Let the size of each file be $L=27$ bits, and every user has cache of size $ML=\frac{NS^{N-1}-q}{KS^{N-1}}L=4$ bits. \\
\textbf{Placement phase: } Let $P=(p_1,p_2,p_3)$ be a random permutation of set $\{1,2,3\}$. The cache content of user $c \in \{1,2,3\}$ is
$Z_c=\{ W_{1,p_c}^t \oplus W_{2,p_c}^t \oplus W_{3,p_c}^t \ | \ t=6,7,8,9 \}.$ \\
\textbf{Delivery phase:} Let the demand vector be $\theta=(2,1,3)$. For each $j \in \{1,2,3\}$, the users generate three permutations of $\{1,2, \ldots, 9\}$ given as
$$a^j \rightarrow (a^j_1, a^j_2, \ldots, a^j_9), \ b^j \rightarrow (b^j_1, b^j_2, \ldots, b^j_9),  \ c^j \rightarrow (c^j_1, c^j_2, \ldots, c^j_9), $$
where $(b^1_6, b^1_7, b^1_8, b^1_9)=(a^2_6,a^2_7,a^2_8,a^2_9) =(c^3_6,c^3_7,c^3_8,c^3_9)=(6,7,8,9)$. Now after running Algorithm \ref{algo2}, we get the queries shown in Table \ref{tb4:ex1} for $j=1,2,3$.
\begin{table}
\caption{ Query tables for Example \ref{ex4.1}\label{tb4:ex1}}
\begin{subtable}[c]{0.5\textwidth}
\centering
\resizebox{\columnwidth}{!}{%
\begin{tabular}{ |c|c|c| } 
 \hline
 DB1 & DB2 & DB3  \\
\hline 
$W_{1,p_1}^{a^1_1}$   & $W_{1,p_1}^{a^1_2} + W_{2,p_1}^{b^1_1}$   & $W_{1,p_1}^{a^1_4} + W_{2,p_1}^{b^1_1}$  \\   
$W_{2,p_1}^{b^1_1}$   & $W_{1,p_1}^{a^1_3} + W_{3,p_1}^{c^1_1}$   & $W_{1,p_1}^{a^1_5} + W_{3,p_1}^{c^1_1}$  \\ 
$W_{3,p_1}^{c^1_1}$   & $W_{2,p_1}^{b^1_2} + W_{1,p_1}^{a^1_1}$   & $W_{2,p_1}^{b^1_4} + W_{1,p_1}^{a^1_1}$  \\

$W_{1,p_1}^{a^1_6} + W_{2,p_1}^{b^1_3} + W_{3,p_1}^{c^1_3}$   & $W_{2,p_1}^{b^1_3} + W_{3,p_1}^{c^1_1}$   & $W_{2,p_1}^{b^1_5} + W_{3,p_1}^{c^1_1}$  \\  
$W_{1,p_1}^{a^1_7} + W_{2,p_1}^{b^1_2} + W_{3,p_1}^{c^1_5}$    & $W_{3,p_1}^{c^1_2} + W_{1,p_1}^{a^1_1}$   & $W_{3,p_1}^{c^1_4} + W_{1,p_1}^{a^1_1}$  \\ 
$ W_{3,p_1}^{c^1_6} + W_{1,p_1}^{a^1_3} + W_{2,p_1}^{b^1_4} $   & $W_{3,p_1}^{c^1_3} + W_{2,p_1}^{b^1_1}$   & $W_{3,p_1}^{c^1_5} + W_{2,p_1}^{b^1_1}$  \\ 
$ W_{3,p_1}^{c^1_7} + W_{1,p_1}^{a^1_5} + W_{2,p_1}^{b^1_5} $  & $W_{1,p_1}^{a^1_8} + W_{2,p_1}^{b^1_5} + W_{3,p_1}^{c^1_5}$  & $W_{1,p_1}^{a^1_9} + W_{2,p_1}^{b^1_2} + W_{3,p_1}^{c^1_3}$ \\
 & $ W_{3,p_1}^{c^1_8} + W_{1,p_1}^{a^1_5} + W_{2,p_1}^{b^1_4} $  & $ W_{3,p_1}^{c^1_9} + W_{1,p_1}^{a^1_3} + W_{2,p_1}^{b^1_3}$ \\

 \hline

\end{tabular}}
\subcaption{For $j=1$}
\end{subtable}
\begin{subtable}[c]{0.5\textwidth}
\centering
\resizebox{\columnwidth}{!}{%
\begin{tabular}{ |c|c|c| } 
 \hline
 DB1 & DB2 & DB3  \\
\hline 
$W_{1,p_2}^{a^2_1}$   & $W_{1,p_2}^{a^2_2} + W_{2,p_2}^{b^2_1}$   & $W_{1,p_2}^{a^2_4} + W_{2,p_2}^{b^2_1}$  \\   
$W_{2,p_2}^{b^2_1}$   & $W_{1,p_2}^{a^2_3} + W_{3,p_2}^{c^2_1}$   & $W_{1,p_2}^{a^2_5} + W_{3,p_2}^{c^2_1}$  \\ 
$W_{3,p_2}^{c^2_1}$   & $W_{2,p_2}^{b^2_2} + W_{1,p_2}^{a^2_1}$   & $W_{2,p_2}^{b^2_4} + W_{1,p_2}^{a^2_1}$  \\

$W_{2,p_2}^{b^2_6} + W_{1,p_2}^{a^2_3} + W_{3,p_2}^{c^2_3}$   & $W_{2,p_2}^{b^2_3} + W_{3,p_2}^{c^2_1}$   & $W_{2,p_2}^{b^2_5} + W_{3,p_2}^{c^2_1}$  \\  
$W_{2,p_2}^{b^2_7} + W_{1,p_2}^{a^2_2} + W_{3,p_2}^{c^2_5}$   & $W_{3,p_2}^{c^2_2} + W_{1,p_2}^{a^2_1}$   & $W_{3,p_2}^{c^2_4} + W_{1,p_2}^{a^2_1}$  \\ 
$ W_{3,p_2}^{c^2_6} + W_{1,p_2}^{a^2_4} + W_{2,p_2}^{b^2_3} $   & $W_{3,p_2}^{c^2_3} + W_{2,p_2}^{b^2_1}$   & $W_{3,p_2}^{c^2_5} + W_{2,p_2}^{b^2_1}$  \\ 

$ W_{3,p_2}^{c^2_7} + W_{1,p_2}^{a^2_5} + W_{2,p_2}^{b^2_5} $  & $W_{2,p_2}^{b^2_8} + W_{1,p_2}^{a^2_5} + W_{3,p_2}^{c^2_5}$ & $W_{2,p_2}^{b^2_9} + W_{1,p_2}^{a^2_2} + W_{3,p_2}^{c^2_3}$ \\
 & $ W_{3,p_2}^{c^2_8} + W_{1,p_2}^{a^2_4} + W_{2,p_2}^{b^2_5} $  & $ W_{3,p_2}^{c^2_9} + W_{1,p_2}^{a^2_3} + W_{2,p_2}^{b^2_3}$ \\

 \hline
\end{tabular} }
\subcaption{For $j=2$}
\end{subtable}
\\
\center{
\begin{subtable}[c]{0.5\textwidth}
\centering
\resizebox{\columnwidth}{!}{%
\begin{tabular}{ |c|c|c| } 
 \hline
 DB1 & DB2 & DB3  \\
\hline 
$W_{1,p_3}^{a^3_1}$   & $W_{1,p_3}^{a^3_2} + W_{2,p_3}^{b^3_1}$   & $W_{1,p_3}^{a^3_4} + W_{2,p_3}^{b^3_1}$  \\   
$W_{2,p_3}^{b^3_1}$   & $W_{1,p_3}^{a^3_3} + W_{3,p_3}^{c^3_1}$   & $W_{1,p_3}^{a^3_5} + W_{3,p_3}^{c^3_1}$  \\ 
$W_{3,p_3}^{c^3_1}$   & $W_{2,p_3}^{b^3_2} + W_{1,p_3}^{a^3_1}$   & $W_{2,p_3}^{b^3_4} + W_{1,p_3}^{a^3_1}$  \\

$W_{1,p_3}^{a^3_6} + W_{2,p_3}^{b^3_3} + W_{3,p_3}^{c^3_3}$   & $W_{2,p_3}^{b^3_3} + W_{3,p_3}^{c^3_1}$   & $W_{2,p_3}^{b^3_5} + W_{3,p_3}^{c^3_1}$  \\  
$W_{1,p_3}^{a^3_7} + W_{2,p_3}^{b^3_5} + W_{3,p_3}^{c^3_2}$    & $W_{3,p_3}^{c^3_2} + W_{1,p_3}^{a^3_1}$   & $W_{3,p_3}^{c^3_4} + W_{1,p_3}^{a^3_1}$  \\ 

$ W_{2,p_3}^{b^3_6} + W_{1,p_3}^{a^3_3} + W_{3,p_3}^{c^3_4} $   & $W_{3,p_3}^{c^3_3} + W_{2,p_3}^{b^3_1}$   & $W_{3,p_3}^{c^3_5} + W_{2,p_3}^{b^3_1}$  \\ 
$ W_{2,p_3}^{b^3_7} + W_{1,p_3}^{a^3_5} + W_{3,p_3}^{c^3_5} $  & $W_{1,p_3}^{a^3_8} + W_{2,p_3}^{b^3_5} + W_{3,p_3}^{c^3_4}$  & $W_{1,p_3}^{a^3_9} + W_{2,p_3}^{b^3_3} + W_{3,p_3}^{c^3_2}$ \\

 & $ W_{2,p_3}^{b^3_8} + W_{1,p_3}^{a^3_5} + W_{3,p_3}^{c^3_5} $  & $ W_{2,p_3}^{b^3_9} + W_{1,p_3}^{a^3_3} + W_{3,p_3}^{c^3_3} $ \\

 \hline
\end{tabular}}
\subcaption{For $j=3$}
\end{subtable}}
\end{table}

From the Table \ref{tb4:ex1}, we get all the subsubfiles of the subfiles $W_{1,p_1}, W_{3,p_1}, W_{2,p_2}, W_{3,p_2},$ $W_{1, p_3}, W_{2,p_3}$ and first $5$ subsubfiles of the subfiles $W_{2,p_1}, W_{1,p_2}, W_{3,p_3}$. The demand of user $1$ is $W_2$ and the user already got two subfiles $W_{2,p_2}, W_{2,p_3}$ completely and first $5$ subsubfiles of $W_{2, p_1}$. Since the cache content of user $1$ is
$Z_1=\{ W_{1,p_1}^t \oplus W_{2,p_1}^t \oplus W_{3,p_1}^t \ | \ t=6,7,8,9 \},$
remaining $4$ subsubfiles of $W_{2, p_1}$ can be obtained with the help of cache content $Z_1$. In a similar manner, the other two users will get their desired files. The rate in this scheme is $R=\frac{q}{S^{N-1}}=\frac{23}{9}=2.556$, while the rate of the PD scheme is $R'=2.769.$
\end{example}

\subsubsection{For $N<K$}\label{subsec4.4}

\begin{algorithm}
\caption{Generates query sets for each database $s \in S$ for the scheme [\ref{subsec3.1}] for $N <K$.\label{algo3}}

\SetAlgoLined
\KwData{ Permutations $P^{\lambda i}, \forall i \in [N], \lambda \in [K]$, demand vector $\theta$}

\KwResult {Query sets $Q_{s},\  \forall s \in [S]$}
\ForEach {$\lambda \in \mathcal{B}$}
{
		QSet1 ($p_{\lambda}; P^{\lambda i}, \ \forall i \in [N] ; d_{\lambda}$);

		$Q^{\lambda}_s \gets \left\{ \sum_{x \in Q} x \ | \ Q \in Q'_s \right\}, \ \forall  \ s \in[S] ;$
}
\ForEach {$\lambda \in [K]\backslash \mathcal{B}$}
{
		QSet2 ($W_{i,p_{\rho^i}}^x + W_{i,p_c}^x , \forall i \in[N] ;$ $P^{\lambda i}, \ \forall i \in [N] $);

		$Q^{\lambda}_s \gets \left\{ \sum_{x \in Q} x \ | \ Q \in Q'_s \right\}, \ \forall  \ s \in[S] ;$
}
		$Q_s \gets \bigcup_{\lambda=1}^K Q^{\lambda}_s, \ \forall  \ s \in[S] ;$

\end{algorithm}

Assuming that each file is demanded by at least one user, consider that the demand vector is $\theta = (d_1,d_2, \ldots, d_K)$. Now choose a set $\mathcal{B} \subseteq [K]$ (call it base set) such that $|\mathcal{B}|=N$ and $\{ d_c \ | \ c \in \mathcal{B}\} = [N]$. For each $c \in \mathcal{B}$, the users generate $N$ permutations $\{P^{c1},P^{c2}, \ldots, P^{cN}\}$ of the set $[S^{N-1}]$ such that
$P^{ci}  = (p^i_1, p^i_2, \ldots, p^i_{S^{N-1}} ) \ \text{for all} \ i \in[N] $ and
$(p^{d_{c}}_{H+1}, \ p^{d_{c}}_{H+2}, \ldots, p^{d_{c}}_{S^{N-1}}) = (H+1, H+2, \ldots, S^{N-1}).$
In Algorithm \ref{algo3}, for each $c \in \mathcal{B}$, Function 1 (QSet1) is called with the following values of the parameters
$$p_c \longrightarrow j; \ \ P^{ci} \longrightarrow a^i, \forall i \in[N]; \ \ d_{c} \longrightarrow d .$$

For each $c \in \mathcal{B}$, we get $S$ query sets for each database as output of Function 1, and denote them as
  $ Q^c_s, \ \text{for all}\ s \in[S].$
 
Now for each $c \in [K] \backslash \mathcal{B}$, the users generate
\begin{itemize}
\item  a one-to-one mapping from $[N]$ to $\mathcal{B}$, say $\rho = (\rho^1, \rho^2, \ldots, \rho^N)$, such that
$d_{\rho^j}=j$ for $j=d_c$, and $d_{\rho^j} \neq j$ for all $j \in [N]$ and $j \neq d_c$.

\item $N$ random permutations $\{P^{c1}, P^{c2}, \ldots, P^{cN}\}$ of the set $[S^{N-1}]$ such that
$
P^{ci}  = (p^i_1, p^i_2, \ldots, p^i_{S^{N-1}} )$ for all $i \in[N].$
\end{itemize}
In Algorithm \ref{algo3}, for each $c \in [K] \backslash \mathcal{B}$, Function 2 (QSet2) is called with the following values of the parameters
\begin{align*}
W_{i,p_{\rho^i}}^x + W_{i,p_c}^x &\longrightarrow \omega_i(x), \forall i \in[N] \\
P^{ci} &\longrightarrow a^i, \forall i \in[N] .
\end{align*}
Again for each $c \in [K] \backslash \mathcal{B}$, we get $S$ query sets for each database as output of Function 2, denoted as 
  $Q^c_s, \ \text{for all}\ s \in[S].$
Therefore, the query set generated by Algorithm \ref{algo3} for database $s$ is
$Q_s=\bigcup_{c=1}^K Q^c_s,$ for all $s\in [S]$

\begin{note*}
For decoding the demand of user $c, c \in [K] \backslash \mathcal{B}$, we need subfiles $ W_{i,p_c}$ for $i \in [N],  i\neq d_c$ and  the subfile $W_{d_c,p_{c'}}+ W_{d_c,p_c}$, where $c' \in \mathcal{B}$ and $d_{c}=d_{c'}$. However, to maintain the symmetry along all $c' \in \mathcal{B}$, we consider a one-to-one mapping $\rho$ from $[N]$ to $\mathcal{B}$, and obtained the subfiles $W_{i,p_{\rho^i}} + W_{i,p_c}$ for all $i \in [N]$ using Function 2.
\end{note*}

\subsubsection*{Decoding}
In this subsection, we prove that after receiving answers to all the $Q_s, s\in [S]$ from each database, all the users will get their desired files.
\begin{enumerate}
\item Consider a user $c \in \mathcal{B}$ with demand of file $W_{d_{c}}$. 
 The file $W_{d_{c}}$ has $K$ subfiles $W_{d_{c},1}, W_{d_{c},2}, \ldots,$ $W_{d_{c},K}$.
 User will directly get all the subsubfiles of subfiles $W_{d_c,p_j}$, where $j \in \mathcal{B}$ and $j \neq c$, as Function 1 is called in Algorithm \ref{algo3} for every $j \in \mathcal{B}$ and each user in $\mathcal{B}$ has distinct demands.
 By calling Function 1 for $c$, first $H$ subsubfiles of subfile $W_{d_c,p_c}$ will be obtained as $(p^{d_{c}}_{H+1}, p^{d_{c}}_{H+2}, \ldots, p^{d_{c}}_{S^{N-1}}) = (H+1, H+2, \ldots, S^{N-1})$. Now remaining $S^{N-1}-H$ subsubfiles of $W_{d_c,p_c}$ can be obtained using the cache content of user $c$, which is 
$$
Z_{c} = \left\{ W_{1,p_{c}}^t \oplus W_{2,p_{c}}^t \oplus \cdots \oplus W_{N,p_{c}}^t \ |  t=H+1, H+2, \ldots, S^{N-1} \right\}, $$
as the user already have all the subsubfiles of $W_{i,p_c}, i \in [N]$ and $i \neq d_c$ by calling Function 1 for $c$.
 Until now user $c$ has obtained subfiles $W_{d_c,p_j}, j \in \mathcal{B}$. The subfiles remains to be obtained are $W_{d_c,p_j}, j \in [K] \backslash \mathcal{B}$.
 In Algorithm \ref{algo3}, by calling Function 2 for $j \in [K] \backslash \mathcal{B}$, we get all subsubfiles of $W_{i,p_{\rho^i}} + W_{i,p_{j}}$, where $i\in [N]$ and $\rho^i \in \mathcal{B}$. For $i=d_c$, we have $W_{d_c,p_{\rho^{d_c}}} + W_{d_c,p_{j}}$. Since we already have all subfiles $W_{d_c,p_{j'}}, j' \in \mathcal{B}$ and $\rho^{d_c} \in \mathcal{B}$, we get $W_{d_c, p_{j}}$.

\item Consider a user $c \in [K] \backslash \mathcal{B}$ with demand of file $W_{d_{c}}$. 
 The file $W_{d_{c}}$ has $K$ subfiles $W_{d_{c},1}, W_{d_{c},2}, \ldots, W_{d_{c},K}$. Clearly, there exist $c' \in \mathcal{B}$ such that $d_c=d_{c'}$. The user $c$ will directly get all the subsubfiles of subfiles $W_{d_c,p_j}$, where $j \in \mathcal{B}$ and $j \neq c'$, as in Algorithm \ref{algo3}, Function 1 is called for every $j \in \mathcal{B}$ and each user in $\mathcal{B}$ has distinct demands.
 By calling Function 1 for $c'$, first $H$ subsubfiles of subfile $W_{d_c,p_{c'}}$ will be obtained, i.e., $W_{d_c,p_{c'}}^1, W_{d_c,p_{c'}}^2, \ldots, W_{d_c,p_{c'}}^H$. By calling Function 2 for $c$, we get $W_{i,p_{\rho^i}} + W_{i,p_c}$, where $i \in [N]$ and $\rho^i \in \mathcal{B}$. Now there are two type of cases,
\begin{enumerate}
\item When $i \neq d_c$, we have $d_{\rho^i} \neq i$. Since $\rho^i \in \mathcal{B}$ and $d_{\rho_i} \neq i$, we have subfile $W_{i,p_{\rho^i}}$ when Function 1 is called in Algorithm \ref{algo3} for $\rho^i$. Therefore, user $c$ gets $W_{i,p_c}$, for all $i \in [N]$ and $i \neq d_c$ (using $W_{i,p_{\rho^i}} + W_{i,p_c}$). Now using the cache content of user $c$, which is 
$$
Z_{c} = \left\{ W_{1,p_{c}}^t \oplus W_{2,p_{c}}^t \oplus \cdots \oplus W_{N,p_{c}}^t \ | \ t=H+1, H+2, \ldots, S^{N-1} \right\}, 
$$
we get $S^{N-1}-H$ subsubfiles of $W_{d_c, p_c}$, i.e., $W_{d_c, p_c}^{H+1}, W_{d_c, p_c}^{H+2}, \ldots, W_{d_c, p_c}^{S^{N-1}}$.

\item When $i=d_c$, we have $W_{d_c,p_{\rho^{d_c}}} + W_{d_c,p_c}$ for $d_{\rho^i}=i$. That means we have $W_{d_c,p_{c'}} + W_{d_c,p_c}$ (as $\mathcal{B}$ contains only one user with demand $d_c$ which is $c'$, i.e., $\rho^{d_c}=\rho^i=c'$).  Since user $c$ already have first $H$ subsubfiles of $W_{d_c,p_{c'}}$ and last $S^{N-1}-H$ subsubfiles of $W_{d_c,p_c}$, the user obtained both subfiles $W_{d_c,p_{c'}}$ and $W_{d_c,p_c}$, completely.
\end{enumerate}
 Until now user $c$ has obtained subfiles $W_{d_c,p_j}, j \in \mathcal{B}$ and $W_{d_c, p_c}$. The subfiles remains to be obtained are $W_{d_c,p_j}, j \in [K] \backslash (\mathcal{B}\cup \{c\})$.  By calling Function 2 in Algorithm \ref{algo3} for $j \in [K] \backslash (\mathcal{B}\cup \{c\})$, we get $W_{i,p_{\rho^i}} + W_{i,p_j}$, where $i \in [N]$ and $\rho^i \in \mathcal{B}$. For $i=d_c$, we have $W_{d_c,p_{\rho^{d_c}}} + W_{d_c,p_j}$. Since user $c$ already has subfile $W_{d_c,p_{\rho^{d_c}}}$, user can obtain subfile $W_{d_c,p_j}$.
\end{enumerate}

\begin{theorem}\label{thm4.2}
The proposed scheme achieves the rate $\frac{N}{K}\left[ \frac{q}{S^{N-1}} +K-N \right]$ for $N<K$.
\end{theorem}
\begin{proof}
Since Function 1 is called for each $c \in \mathcal{B}$, and Function 2 is called for each  $c' \in [K] \backslash \mathcal{B}$ in Algorithm \ref{algo3}, the total number of queries is $Nq'+(K-N)q''$, where $q'$ and $q''$ are the numbers of queries generated in Function 1 and Function 2, respectively. From Theorem \ref{thm4.1}, we know that $q'=q$. The number of queries generated in Function 2 is
\begin{align*}
q'' &= \sum_{s=1}^S Q'_s = N+\sum_{s=1}^S \sum_{k=2}^N (k \phi_s(S,k)) {N \choose k} \\
&= \sum_{s=1}^S \sum_{k=1}^N k {N \choose k} \phi_s(S,k)=N \sum_{s=1}^S \sum_{k=1}^N  {N-1 \choose {k-1}} \phi_s(S,k)=NS^{N-1}.
\end{align*}
Hence, the total number of queries is $Nq+N(K-N)S^{N-1}$. Since the answer for each query is of the size of a subsubfile, we have
$$R=\frac{Nq+N(K-N)S^{N-1}}{KS^{N-1}}=\frac{N}{K}\left[ \frac{q}{S^{N-1}} +K-N \right].$$
\end{proof}

\begin{example}\label{ex4.2}
Consider a PIR problem with $S=3$ servers, $N=3$ files and $K=5$ users. Each file is divided into $5$ subfiles, i.e.,
$W_i  \rightarrow W_{i,1}, W_{i,2}, W_{i,3}, W_{i,4}, W_{i,5}, \ \ \forall i \in \{1,2,3\}.$
Each subfile is further divided into $9$ subsubfiles, i.e.,
$W_{i,j} \rightarrow W_{i,j}^1, W_{i,j}^2, \ldots, W_{i,j}^9, \ \forall i \in \{1,2,3\}, j \in \{ 1,2,3,4,5\}.$
Now we compute the values of function $\psi_s(S,k)$, where $s, k\in \{1,2,3\}$ which is same as in Example \ref{ex4.1}. So we have $q=23$, and $H=q-(N-1)S^{N-1}=5$.
Let the size of each file be $L=45$ bits, and every user has cache of size $ML=\frac{NS^{N-1}-q}{KS^{N-1}}L=4$ bits. \\
\textbf{Placement phase: } Let $P=(p_1, p_2, p_3, p_4, p_5)$ be a random permutation of set $\{1,2,3,4,5\}$. The cache content of user $c \in \{1,2,3,4,5\}$ is
$Z_c=\{ W_{1,p_c}^t \oplus W_{2,p_c}^t \oplus W_{3,p_c}^t \ | \ t=6,7,8,9 \}.$ \\
\textbf{Delivery phase:} Let the demand vector be $\theta=(2,3,2,1,3)$. Consider a base set $\mathcal{B}=\{1,2,4\}$. For each $j \in \mathcal{B}$, consider three permutations of $\{1,2, \ldots, 9\}$ generated by the users and given as
$$a^j \rightarrow (a^j_1, a^j_2, \ldots, a^j_9), b^j \rightarrow (b^j_1, b^j_2, \ldots, b^j_9), c^j \rightarrow (c^j_1, c^j_2, \ldots, c^j_9) ,$$
where $(b^1_6, b^1_7, b^1_8, b^1_9)=(c^2_6,c^2_7,c^2_8,c^2_9) =(a^4_6,a^4_7,a^4_8,a^4_9)=(6,7,8,9)$. Now after running Algorithm \ref{algo3}, we get the queries shown in Table \ref{tb7:ex2} for $j=1,2,4$.

\begin{table}
\caption{Query tables for Example \ref{ex4.2} \label{tb7:ex2}}
\begin{subtable}[c]{0.5\textwidth}
\centering
\resizebox{\columnwidth}{!}{%
\begin{tabular}[\textwidth]{ |c|c|c| } 
 \hline
 DB1 & DB2 & DB3  \\
\hline 
$W_{1,p_1}^{a^1_1}$   & $W_{1,p_1}^{a^1_2} + W_{2,p_1}^{b^1_1}$   & $W_{1,p_1}^{a^1_4} + W_{2,p_1}^{b^1_1}$  \\   
$W_{2,p_1}^{b^1_1}$   & $W_{1,p_1}^{a^1_3} + W_{3,p_1}^{c^1_1}$   & $W_{1,p_1}^{a^1_5} + W_{3,p_1}^{c^1_1}$  \\ 
$W_{3,p_1}^{c^1_1}$   & $W_{2,p_1}^{b^1_2} + W_{1,p_1}^{a^1_1}$   & $W_{2,p_1}^{b^1_4} + W_{1,p_1}^{a^1_1}$  \\

$W_{1,p_1}^{a^1_6} + W_{2,p_1}^{b^1_3} + W_{3,p_1}^{c^1_3}$   & $W_{2,p_1}^{b^1_3} + W_{3,p_1}^{c^1_1}$   & $W_{2,p_1}^{b^1_5} + W_{3,p_1}^{c^1_1}$  \\  
$W_{1,p_1}^{a^1_7} + W_{2,p_1}^{b^1_2} + W_{3,p_1}^{c^1_5}$    & $W_{3,p_1}^{c^1_2} + W_{1,p_1}^{a^1_1}$   & $W_{3,p_1}^{c^1_4} + W_{1,p_1}^{a^1_1}$  \\ 
$ W_{3,p_1}^{c^1_6} + W_{1,p_1}^{a^1_3} + W_{2,p_1}^{b^1_4} $   & $W_{3,p_1}^{c^1_3} + W_{2,p_1}^{b^1_1}$   & $W_{3,p_1}^{c^1_5} + W_{2,p_1}^{b^1_1}$  \\ 
$ W_{3,p_1}^{c^1_7} + W_{1,p_1}^{a^1_5} + W_{2,p_1}^{b^1_5} $  & $W_{1,p_1}^{a^1_8} + W_{2,p_1}^{b^1_5} + W_{3,p_1}^{c^1_5}$  & $W_{1,p_1}^{a^1_9} + W_{2,p_1}^{b^1_2} + W_{3,p_1}^{c^1_3}$ \\
 & $ W_{3,p_1}^{c^1_8} + W_{1,p_1}^{a^1_5} + W_{2,p_1}^{b^1_4} $  & $ W_{3,p_1}^{c^1_9} + W_{1,p_1}^{a^1_3} + W_{2,p_1}^{b^1_3}$ \\

 \hline

\end{tabular}}
\subcaption{For $j=1$}
\end{subtable}
\begin{subtable}[c]{0.5\textwidth}
\centering
\resizebox{\columnwidth}{!}{%
\begin{tabular}[\textwidth]{ |c|c|c| } 
 \hline
 DB1 & DB2 & DB3  \\
\hline 
$W_{1,p_2}^{a^2_1}$   & $W_{1,p_2}^{a^2_2} + W_{2,p_2}^{b^2_1}$   & $W_{1,p_2}^{a^2_4} + W_{2,p_2}^{b^2_1}$  \\   
$W_{2,p_2}^{b^2_1}$   & $W_{1,p_2}^{a^2_3} + W_{3,p_2}^{c^2_1}$   & $W_{1,p_2}^{a^2_5} + W_{3,p_2}^{c^2_1}$  \\ 
$W_{3,p_2}^{c^2_1}$   & $W_{2,p_2}^{b^2_2} + W_{1,p_2}^{a^2_1}$   & $W_{2,p_2}^{b^2_4} + W_{1,p_2}^{a^2_1}$  \\

$ W_{1,p_2}^{a^2_6} +  W_{2,p_2}^{b^2_3} + W_{3,p_2}^{c^2_2}$   & $W_{2,p_2}^{b^2_3} + W_{3,p_2}^{c^2_1}$   & $W_{2,p_2}^{b^2_5} + W_{3,p_2}^{c^2_1}$  \\  
$W_{1,p_2}^{a^2_7} + W_{2,p_2}^{b^2_5} +  W_{3,p_2}^{c^2_4}$   & $W_{3,p_2}^{c^2_2} + W_{1,p_2}^{a^2_1}$   & $W_{3,p_2}^{c^2_4} + W_{1,p_2}^{a^2_1}$  \\ 
$ W_{2,p_2}^{b^2_6}+ W_{1,p_2}^{a^2_3} + W_{3,p_2}^{c^2_3}   $   & $W_{3,p_2}^{c^2_3} + W_{2,p_2}^{b^2_1}$   & $W_{3,p_2}^{c^2_5} + W_{2,p_2}^{b^2_1}$  \\ 

$ W_{2,p_2}^{b^2_7}+ W_{1,p_2}^{a^2_5} + W_{3,p_2}^{c^2_5} $  & $ W_{1,p_2}^{a^2_8} +  W_{2,p_2}^{b^2_5} + W_{3,p_2}^{c^2_4}$ & $ W_{1,p_2}^{a^2_9} +  W_{2,p_2}^{b^2_3} + W_{3,p_2}^{c^2_2}$ \\
 & $ W_{2,p_2}^{b^2_8}+ W_{1,p_2}^{a^2_5} + W_{3,p_2}^{c^2_5} $  & $ W_{2,p_2}^{b^2_9}+ W_{1,p_2}^{a^2_3} + W_{3,p_2}^{c^2_3}$ \\

 \hline

\end{tabular}}
\subcaption{For $j=2$}
\end{subtable}
\\
\center{
\begin{subtable}[c]{0.5\textwidth}
\centering
\resizebox{\columnwidth}{!}{%

\begin{tabular}[\textwidth]{ |c|c|c| } 
 \hline
 DB1 & DB2 & DB3  \\
\hline 
$W_{1,p_4}^{a^4_1}$   & $W_{1,p_4}^{a^4_2} + W_{2,p_4}^{b^4_1}$   & $W_{1,p_4}^{a^4_4} + W_{2,p_4}^{b^4_1}$  \\   
$W_{2,p_4}^{b^4_1}$   & $W_{1,p_4}^{a^4_3} + W_{3,p_4}^{c^4_1}$   & $W_{1,p_4}^{a^4_5} + W_{3,p_4}^{c^4_1}$  \\ 
$W_{3,p_4}^{c^4_1}$   & $W_{2,p_4}^{b^4_2} + W_{1,p_4}^{a^4_1}$   & $W_{2,p_4}^{b^4_4} + W_{1,p_4}^{a^4_1}$  \\

$W_{2,p_4}^{b^4_6} + W_{1,p_4}^{a^4_2} + W_{3,p_4}^{c^4_3}$   & $W_{2,p_4}^{b^4_3} + W_{3,p_4}^{c^4_1}$   & $W_{2,p_4}^{b^4_5} + W_{3,p_4}^{c^4_1}$  \\  
$W_{2,p_4}^{b^4_7} + W_{1,p_4}^{a^4_4} + W_{3,p_4}^{c^4_5}$    & $W_{3,p_4}^{c^4_2} + W_{1,p_4}^{a^4_1}$   & $W_{3,p_4}^{c^4_4} + W_{1,p_4}^{a^4_1}$  \\ 

$ W_{3,p_4}^{c^4_6} +  W_{1,p_4}^{a^4_3} + W_{2,p_4}^{b^4_3}  $   & $W_{3,p_4}^{c^4_3} + W_{2,p_4}^{b^4_1}$   & $W_{3,p_4}^{c^4_5} + W_{2,p_4}^{b^4_1}$  \\ 
$ W_{3,p_4}^{c^4_7} +  W_{1,p_4}^{a^4_5} + W_{2,p_4}^{b^4_5}   $  & $W_{2,p_4}^{b^4_8} + W_{1,p_4}^{a^4_4} + W_{3,p_4}^{c^4_5}$  & $W_{2,p_4}^{b^4_9} + W_{1,p_4}^{a^4_2} + W_{3,p_4}^{c^4_3}$ \\

 & $ W_{3,p_4}^{c^4_8} +  W_{1,p_4}^{a^4_5} + W_{2,p_4}^{b^4_5}   $  & $W_{3,p_4}^{c^4_9} +  W_{1,p_4}^{a^4_3} + W_{2,p_4}^{b^4_3}  $ \\

 \hline

\end{tabular}}
\subcaption{For $j=4$}
\end{subtable}}
\end{table}

For each $j \in [5]\backslash \mathcal{B}=\{3,5\}$, the users generate 
\begin{itemize}
\item a one-to-one mapping from $\{ 1, 2, 3 \}$ to $\mathcal{B}$, say $\rho = (\rho^1, \rho^2, \rho^3)$, such that
$d_{\rho^i}=i$ for $i=d_j$, and $d_{\rho^i} \neq i$ for all $i \in \{ 1,2, 3\}$ and $i \neq d_j$. For this example, 
$\rho_3 = (2,1,4)$ for $j=3$, and $\rho_5 = (1,4,2)$ for $j=5$.

\item three random permutations of the set $\{1,2, \ldots, 9\}$ as
$$
a^j \rightarrow (a^j_1, a^j_2, \ldots, a^j_9), 
b^j \rightarrow (b^j_1, b^j_2, \ldots, b^j_9), 
c^j \rightarrow (c^j_1, c^j_2, \ldots, c^j_9) .
$$
\end{itemize}

By running Algorithm \ref{algo3}, we get the queries shown in Table \ref{tb10:ex2} for $j \in \{3,5\}$.

\begin{table}
\caption{Query tables for Example \ref{ex4.2} \label{tb10:ex2}}
\center{
\begin{subtable}[c]{0.7\textwidth}
\centering
\resizebox{\columnwidth}{!}{%
\begin{tabular}[\textwidth]{ |c|c|c| } 
 \hline
 DB1 & DB2 & DB3  \\
\hline 
$W_{1,p_2}^{a^3_1} + W_{1,p_3}^{a^3_1}$   & $W_{1,p_2}^{a^3_2} + W_{1,p_3}^{a^3_2} + W_{2,p_1}^{b^3_1}+ W_{2,p_3}^{b^3_1}$   & $W_{1,p_2}^{a^3_4} + W_{1,p_3}^{a^3_4} + W_{2,p_1}^{b^3_1}+ W_{2,p_3}^{b^3_1}$  \\   

$W_{2,p_1}^{b^3_1}+ W_{2,p_3}^{b^3_1}$   & $W_{1,p_2}^{a^3_3} + W_{1,p_3}^{a^3_3}+ W_{3,p_4}^{c^3_1} + W_{3,p_3}^{c^3_1}$   & $W_{1,p_2}^{a^3_5} + W_{1,p_3}^{a^3_5} + W_{3,p_4}^{c^3_1}+ W_{3,p_3}^{c^3_1}$  \\ 

$W_{3,p_4}^{c^3_1}+W_{3,p_3}^{c^3_1}$   & $W_{2,p_1}^{b^3_2} + W_{2,p_3}^{b^3_2} + W_{1,p_2}^{a^3_1}+ W_{1,p_3}^{a^3_1}$   & $W_{2,p_1}^{b^3_4} +W_{2,p_3}^{b^3_4} + W_{1,p_2}^{a^3_1}+ W_{1,p_3}^{a^3_1}$  \\

\makecell{$ W_{1,p_2}^{a^3_6} + W_{1,p_3}^{a^3_6} + W_{2,p_1}^{b^3_2}$ \\$ + W_{2,p_3}^{b^3_2} + W_{3,p_4}^{c^3_2}+ W_{3,p_3}^{c^3_2}$ \\  $ $}   & $W_{2,p_1}^{b^3_3} + W_{2,p_3}^{b^3_3} + W_{3,p_4}^{c^3_1} + W_{3,p_3}^{c^3_1}$   & $W_{2,p_1}^{b^3_5} +W_{2,p_3}^{b^3_5} + W_{3,p_4}^{c^3_1}+ W_{3,p_3}^{c^3_1}$  \\  

\makecell{$ W_{1,p_2}^{a^3_7} + W_{1,p_3}^{a^3_7} + W_{2,p_1}^{b^3_3}$ \\$ + W_{2,p_3}^{b^3_3} + W_{3,p_4}^{c^3_3}+ W_{3,p_3}^{c^3_3}$\\  $ $ }  & $W_{3,p_4}^{c^3_2} + W_{3,p_3}^{c^3_2} + W_{1,p_2}^{a^3_1}+ W_{1,p_3}^{a^3_1}$   & $W_{3,p_4}^{c^3_4} +W_{3,p_3}^{c^3_4} + W_{1,p_2}^{a^3_1}+ W_{1,p_3}^{a^3_1}$  \\ 

\makecell{$ W_{2,p_1}^{b^3_6} + W_{2,p_3}^{b^3_6} + W_{1,p_2}^{a^3_2} $ \\$ + W_{1,p_3}^{a^3_2} + W_{3,p_4}^{c^3_4}+ W_{3,p_3}^{c^3_4}$ \\  $ $}    & $W_{3,p_4}^{c^3_3} +W_{3,p_3}^{c^3_3} + W_{2,p_1}^{b^3_1}+ W_{2,p_3}^{b^3_1}$   & $W_{3,p_4}^{c^3_5} +W_{3,p_3}^{c^3_5} + W_{2,p_1}^{b^3_1}+ W_{2,p_3}^{b^3_1}$  \\ 

\makecell{$ W_{2,p_1}^{b^3_7} + W_{2,p_3}^{b^3_7} + W_{1,p_2}^{a^3_3} $ \\$ + W_{1,p_3}^{a^3_3} + W_{3,p_4}^{c^3_5}+ W_{3,p_3}^{c^3_5}$ \\  $ $}     & \makecell{$ W_{1,p_2}^{a^3_8} + W_{1,p_3}^{a^3_8} + W_{2,p_1}^{b^3_4} + W_{2,p_3}^{b^3_4} $ \\$+ W_{3,p_4}^{c^3_4}+ W_{3,p_3}^{c^3_4}$ \\  $ $}      & \makecell{$ W_{1,p_2}^{a^3_9} + W_{1,p_3}^{a^3_9} + W_{2,p_1}^{b^3_2}+ W_{2,p_3}^{b^3_2} $ \\$ + W_{3,p_4}^{c^3_2}+ W_{3,p_3}^{c^3_2}$ \\  $ $} \\

\makecell{$ W_{3,p_4}^{c^3_6}+ W_{3,p_3}^{c^3_6} + W_{1,p_2}^{a^3_4} $ \\$ + W_{1,p_3}^{a^3_4} + W_{2,p_1}^{b^3_4}+ W_{2,p_3}^{b^3_4} $\\  $ $ }       &    \makecell{$ W_{2,p_1}^{b^3_8} + W_{2,p_3}^{b^3_8} + W_{1,p_2}^{a^3_4}  + W_{1,p_3}^{a^3_4} $ \\$+ W_{3,p_4}^{c^3_5}+ W_{3,p_3}^{c^3_5}$ \\  $ $}    &    \makecell{$ W_{2,p_1}^{b^3_9} + W_{2,p_3}^{b^3_9} + W_{1,p_2}^{a^3_2}  + W_{1,p_3}^{a^3_2}$ \\$ + W_{3,p_4}^{c^3_3}+ W_{3,p_3}^{c^3_3}$ \\  $ $}   \\

\makecell{$ W_{3,p_4}^{c^3_7}+ W_{3,p_3}^{c^3_7} + W_{1,p_2}^{a^3_5} $ \\$ + W_{1,p_3}^{a^3_5} + W_{2,p_1}^{b^3_5}+ W_{2,p_3}^{b^3_5} $ }   &     \makecell{$ W_{3,p_4}^{c^3_8}+ W_{3,p_3}^{c^3_8} + W_{1,p_2}^{a^3_5}  + W_{1,p_3}^{a^3_5} $ \\$+ W_{2,p_1}^{b^3_5}+ W_{2,p_3}^{b^3_5} $\\  $ $ }   &  \makecell{$ W_{3,p_4}^{c^3_9}+ W_{3,p_3}^{c^3_9} + W_{1,p_2}^{a^3_3} + W_{1,p_3}^{a^3_3} $ \\$ + W_{2,p_1}^{b^3_3}+ W_{2,p_3}^{b^3_3} $\\  $ $ }  \\

 \hline
\end{tabular}
}
\subcaption{For $j=3$}
\end{subtable}}
\\
\center{
\begin{subtable}[c]{0.7\textwidth}
\centering
\resizebox{\columnwidth}{!}{%
\begin{tabular}[\textwidth]{ |c|c|c| } 
 \hline
 DB1 & DB2 & DB3  \\
\hline 
$W_{1,p_1}^{a^5_1} + W_{1,p_5}^{a^5_1}$   & $W_{1,p_1}^{a^5_2} + W_{1,p_5}^{a^5_2} + W_{2,p_4}^{b^5_1}+ W_{2,p_5}^{b^5_1}$   & $W_{1,p_1}^{a^5_4} + W_{1,p_5}^{a^5_4} + W_{2,p_4}^{b^5_1}+ W_{2,p_5}^{b^5_1}$  \\   

$W_{2,p_4}^{b^5_1}+ W_{2,p_5}^{b^5_1}$   & $W_{1,p_1}^{a^5_3} + W_{1,p_5}^{a^5_3}+ W_{3,p_2}^{c^5_1} + W_{3,p_5}^{c^5_1}$   & $W_{1,p_1}^{a^5_5} + W_{1,p_5}^{a^5_5} + W_{3,p_2}^{c^5_1}+ W_{3,p_5}^{c^5_1}$  \\ 

$W_{3,p_2}^{c^5_1}+ W_{3,p_5}^{c^5_1}$   & $W_{2,p_4}^{b^5_2} + W_{2,p_5}^{b^5_2} + W_{1,p_1}^{a^5_1}+ W_{1,p_5}^{a^5_1}$   & $W_{2,p_4}^{b^5_4} + W_{2,p_5}^{b^5_4} + W_{1,p_1}^{a^5_1}+ W_{1,p_5}^{a^5_1}$  \\

\makecell{$ W_{1,p_1}^{a^5_6} + W_{1,p_5}^{a^5_6} + W_{2,p_4}^{b^5_2}$ \\$ + W_{2,p_5}^{b^5_2} + W_{3,p_2}^{c^5_2}+ W_{3,p_5}^{c^5_2}$ \\  $ $}   & $W_{2,p_4}^{b^5_3} + W_{2,p_5}^{b^5_3} + W_{3,p_2}^{c^5_1} + W_{3,p_5}^{c^5_1}$   & $W_{2,p_4}^{b^5_5} + W_{2,p_5}^{b^5_5} + W_{3,p_2}^{c^5_1}+ W_{3,p_5}^{c^5_1}$  \\  

\makecell{$ W_{1,p_1}^{a^5_7} + W_{1,p_5}^{a^5_7} + W_{2,p_4}^{b^5_3}$ \\$ + W_{2,p_5}^{b^5_3} + W_{3,p_2}^{c^5_3}+ W_{3,p_5}^{c^5_3}$\\  $ $ }  & $W_{3,p_2}^{c^5_2} + W_{3,p_5}^{c^5_2} + W_{1,p_1}^{a^5_1}+ W_{1,p_5}^{a^5_1}$   & $W_{3,p_2}^{c^5_4} + W_{3,p_5}^{c^5_4} + W_{1,p_1}^{a^5_1}+ W_{1,p_5}^{a^5_1}$  \\ 

\makecell{$ W_{2,p_4}^{b^5_6} + W_{2,p_5}^{b^5_6} + W_{1,p_1}^{a^5_2} $ \\$ + W_{1,p_5}^{a^5_2} + W_{3,p_2}^{c^5_4}+ W_{3,p_5}^{c^5_4}$ \\  $ $}    & $W_{3,p_2}^{c^5_3} + W_{3,p_5}^{c^5_3} + W_{2,p_4}^{b^5_1}+ W_{2,p_5}^{b^5_1}$   & $W_{3,p_2}^{c^5_5} + W_{3,p_5}^{c^5_5} + W_{2,p_4}^{b^5_1}+ W_{2,p_5}^{b^5_1}$  \\ 

\makecell{$ W_{2,p_4}^{b^5_7} + W_{2,p_5}^{b^5_7} + W_{1,p_1}^{a^5_3} $ \\$ + W_{1,p_5}^{a^5_3} + W_{3,p_2}^{c^5_5}+ W_{3,p_5}^{c^5_5}$ \\  $ $}     & \makecell{$ W_{1,p_1}^{a^5_8} + W_{1,p_5}^{a^5_8} + W_{2,p_4}^{b^5_4} + W_{2,p_5}^{b^5_4} $ \\$+ W_{3,p_2}^{c^5_4}+ W_{3,p_5}^{c^5_4}$ \\  $ $}      & \makecell{$ W_{1,p_1}^{a^5_9} + W_{1,p_5}^{a^5_9} + W_{2,p_4}^{b^5_2}+ W_{2,p_5}^{b^5_2} $ \\$ + W_{3,p_2}^{c^5_2}+ W_{3,p_5}^{c^5_2}$ \\  $ $} \\

\makecell{$ W_{3,p_2}^{c^5_6}+ W_{3,p_5}^{c^5_6} + W_{1,p_1}^{a^5_4} $ \\$ + W_{1,p_5}^{a^5_4} + W_{2,p_4}^{b^5_4}+ W_{2,p_5}^{b^5_4} $\\  $ $ }       &    \makecell{$ W_{2,p_4}^{b^5_8} + W_{2,p_5}^{b^5_8} + W_{1,p_1}^{a^5_4}  + W_{1,p_5}^{a^5_4} $ \\$+ W_{3,p_2}^{c^5_5}+ W_{3,p_5}^{c^5_5}$ \\  $ $}    &    \makecell{$ W_{2,p_4}^{b^5_9} + W_{2,p_5}^{b^5_9} + W_{1,p_1}^{a^5_2}  + W_{1,p_5}^{a^5_2}$ \\$ + W_{3,p_2}^{c^5_3}+ W_{3,p_5}^{c^5_3}$ \\  $ $}   \\

\makecell{$ W_{3,p_2}^{c^5_7}+ W_{3,p_5}^{c^5_7} + W_{1,p_1}^{a^5_5} $ \\$ + W_{1,p_5}^{a^5_5} + W_{2,p_4}^{b^5_5}+ W_{2,p_5}^{b^5_5} $ }   &     \makecell{$ W_{3,p_2}^{c^5_8}+ W_{3,p_5}^{c^5_8} + W_{1,p_1}^{a^5_5}  + W_{1,p_5}^{a^5_5} $ \\$+ W_{2,p_4}^{b^5_5}+ W_{2,p_5}^{b^5_5} $\\  $ $ }   &  \makecell{$ W_{3,p_2}^{c^5_9}+ W_{3,p_5}^{c^5_9} + W_{1,p_1}^{a^5_3} + W_{1,p_5}^{a^5_3} $ \\$ + W_{2,p_4}^{b^5_3}+ W_{2,p_5}^{b^5_3} $\\  $ $ }  \\

 \hline
\end{tabular}}
\subcaption{For $j=5$}
\end{subtable}}
\end{table}

From the Table \ref{tb7:ex2}, we get all the subsubfiles of the subfiles $W_{1,p_1}, W_{3,p_1}, W_{1,p_2}, W_{2,p_2},$ $W_{2, p_4}, W_{3,p_4}$ and first $5$ subsubfiles of the subfiles $W_{2,p_1}, W_{3,p_2}, W_{1,p_4}$. From the Table \ref{tb10:ex2}, we get all the subsubfiles of $W_{1,p_2}+ W_{1,p_3}, W_{2,p_1}+ W_{2,p_3},$ $W_{3, p_4} + W_{3,p_3}$, $W_{1, p_1} + W_{1,p_5}$, $W_{2, p_4} + W_{2,p_5}$ and $W_{3, p_2} + W_{3,p_5}$.

The demand of user $1$ is $W_2$ and user already got two subfiles $W_{2,p_2}, W_{2,p_4}$ completely and first $5$ subsubfiles of $W_{2, p_1}$. Since the cache content of user $1$ is
$Z_1=\{ W_{1,p_1}^t \oplus W_{2,p_1}^t \oplus W_{3,p_1}^t \ | \ t=6,7,8,9 \},$
remaining $4$ subsubfiles of $W_{2, p_1}$ can be obtained with the help of cache content $Z_1$. Hence user $1$ obtained $W_{2,p_2}, W_{2,p_4}$ and $W_{2, p_1}$ completely, and remaining two subfiles $W_{2,p_3}, W_{2,p_5}$, can be obtained from $W_{2,p_1}+ W_{2,p_3}$ and $W_{2, p_4} + W_{2,p_5}$, respectively. In a similar manner, user $2$ and user $4$ will get their desired files.

The demand of user $3$ is also $W_2$ and user $3$ already got two subfiles $W_{2,p_2}, W_{2,p_4}$ completely and first $5$ subsubfiles of $W_{2, p_1}$. User $3$ will obtain last $4$ subsubfiles of $W_{2,p_3}$ by using its cache content $Z_3$ and the subfiles $W_{1,p_3}, W_{3,p_3}$. Since the user have first $5$ subsubfiles of $W_{2, p_1}$ and last $4$ subsubfiles of $W_{2,p_3}$, using $W_{2,p_1}+ W_{2,p_3}$, the user will get both the subfiles $W_{2, p_1}$ and $W_{2, p_3}$ completely. Now user $3$ will obtain last remaining subfile $W_{2, p_5}$ using $W_{2, p_4} + W_{2,p_5}$. In a similar manner, user $5$ will get the desired file $W_3$. The rate of this scheme is $R=\frac{N}{K} \left( \frac{q}{S^{N-1}} + K-N \right) =\frac{41}{15} = 2.733$, while rate of the PD scheme for this example is $R' = 2.926$.
\end{example}

\begin{theorem}
The proposed scheme is private.
\end{theorem}
\begin{proof}
The query set for database $s \in [S]$ is
$Q_s=\cup_{j=1}^K Q_s^j,$
where $Q_s^j$ is either generated by Function 1 or Function 2 for all $j \in [K]$. Let $\overline{W}_i$ represent the collection of all subsubfiles of file $W_i$ which appeared in $Q_s$ for all $ i \in [N]$.
\begin{itemize}
\item \textbf{For $N=K$:} In Algorithm \ref{algo2}, Function 1 is used to generate the query sets $Q_s^j$ for all $j\in [K]$ and $s \in [S]$ in which the query structure is symmetric along all the files $W_i, i \in [N]$ as the query structure is generated with the help of the collection $\mathcal{C} = \{ U\in [N] \ | \ |U| =k \}$ used $\psi_s(S,k)$ times, where $k \in [N]$. Also, the number of repetitions of subsubfiles of every file $W_i$ is also symmetric in $Q_s^j$ for all $j \in [K]$. Therefore, the query set $Q_s$ can be completely determined by the collections $\overline{W}_i$, $i \in [N]$.

Let $\theta=(d_1,d_2,\ldots,d_K)$ be the demand vector. In the proposed scheme, for $j \in [K],$ while constructing $Q_s^j$, we have taken $N$ permutations of $[S^{N-1}]$,
$
P^{ji} = (p^i_1, p^i_2, \ldots, p^i_{S^{N-1}} )$ for all  $i \in[N] $
such that
$(p^{d_{j}}_{H+1}, p^{d_{j}}_{H+2}, \ldots, p^{d_{j}}_{S^{N-1}}) = (H+1, H+2, \ldots, S^{N-1}).$
Hence there are exactly $(S^{N-1}-H)$ subsubfiles of $W_{d_j}$ which will not appear in $Q_{s}^j$. Since this is true for each $j \in [K]$ and $Q_s=\cup_{j=1}^K Q_s^j$, there are exactly $(S^{N-1}-H)$ subsubfiles of $W_{d_j}$ which will never appear in $Q_s$. As we have assumed that all the users have distinct demands, we have $\{ d_1, d_2, \ldots, d_K \}=[N]$. That means $(S^{N-1}-H)$ subsubfiles of every file $W_{i}, i \in [N]$ will never appear in $Q_s$ for every $s \in [S]$. This is true for any demand vector $\theta$ and the probability,
$$P(\bm{Q_s}=Q_s \ | \ \bm{\theta}=\theta) = \frac{1}{[ (\mathcal{P}_{\cancel{d}})^{N-1} (\mathcal{P}_d) ] ^K}$$
where $\mathcal{P}_{\cancel{d}}=S^{N-1} (S^{N-1}-1) \cdots (S^{N-1}-t+1)$ and $\mathcal{P}_d =(S^{N-1}-(S^{N-1}-H)) (S^{N-1} - (S^{N-1} -H) -1) \cdots (S^{N-1} - (S^{N-1}-H) -t+1) =H(H-1)(H-2)\cdots (H-t+1),$
and $t$ denotes the number of subsubfiles of a file $W_i$ that appeared in $Q^{\lambda}_s$ generated as an output of Function 1 for some $\lambda \in [K]$. The probability $P(\bm{Q_s}=Q_s \ | \ \bm{\theta}=\theta)$ is independent of the choice of demand vector $\theta$. Therefore, $P(\bm{Q_s}=Q_s \ | \ \bm{\theta}=\theta)=P(\bm{Q_s}=Q_s)$.

\begin{note*}
In the expression of $P(\bm{Q_s}=Q_s \ | \ \bm{\theta}=\theta)$, $H$ should be greater than $t$. The total number of subsubfiles of a file $W_i$ appeared in all $Q_s^j$, $s\in [S]$ is
$$\sum_{s=1}^S \sum_{k=1}^N \mathcal{F}_s^{i,d_j}(k) = \begin{cases} S^{N-1} & \text{for}\ i \neq d_j \\ H & \text{for} \ i =d_j \end{cases}.$$
Therefore, $t \leq \sum_{s=1}^S \sum_{k=1}^N \mathcal{F}_s^{i,d_j}(k) \leq H$.
\end{note*}

Since $P(\bm{Q_s}=Q_s \ | \ \bm{\theta}=\theta)=P(\bm{Q_s}=Q_s)$, we have
$I(\theta  ; Q_s) = P(Q_S) - P(Q_s \ | \ \theta) =0.$
To prove the privacy of demand with respect to the databases, we need to prove that 
$$I(\theta; Q_s, A_s, W_{1:N}, Z_{1:K}) =0,$$
where $Q_s$ is the query set for database $s,$  
$A_s$ is the set of answers to the query set $Q_s,$
$W_{1:N}$ the collection of all $N$ files, and   
$Z_{1:K}$ is the cache content of all $K$ users.
As we have proved, $I(\theta ; Q_s )=0$. Now, $I(\theta ; W_{1:N} \ | \ Q_s) = 0$ because $\theta$ and $Q_s$ are generated by users without using the knowledge of files $W_{1:N}$, $I(\theta ; A_s \ | \ Q_s, W_{1:N}) =0$  because $A_s$ is a deterministic function of $Q_s$ and $W_{1:N}$, and $I(\theta; Z_{1:K} \ | \ Q_s, W_{1:N}, A_s)=0$ because the demand vector is decided by the users independent of $Z_{1:K}$ and databases do not know which user has which part of whole cache content $\cup_{j=1}^K Z_j$. Therefore,
\begin{multline*}
I(\theta; Q_s, A_s, W_{1:N}, Z_{1:K}) = I(\theta ; Q_s ) + I(\theta ; W_{1:N} \ | \ Q_s) + I(\theta ; A_s \ | \ Q_s, W_{1:N}) \\ + I(\theta; Z_{1:K} \ | \ Q_s, W_{1:N}, A_s)=0. 
\end{multline*}

\item \textbf{For $N<K$:} In this case, Function 1 is used for each $j\in \mathcal{B}$ ($\mathcal{B}$ is defined in Subsection \ref{subsec4.4}) and Function 2 is used for each $j \in [K] \backslash \mathcal{B}$. Privacy is maintained in all $Q_s^j, j\in \mathcal{B}$, which can be proved in the same way as in the above case. For each $j \in [K] \backslash \mathcal{B}$, Function 2 is called with parameters $W_{i,p_{\rho^i}}^x + W_{i,p_j}^x \longrightarrow \omega_i(x), \forall i \in[N] $ and $ P^{ji} \longrightarrow a^i, \forall i \in[N]$,
where all $P^{ji}$'s are random permutations of $[S^{N-1}]$. Since the input of Function 2 is not dependent on the demand vector and all $P^{ji}$'s are random permutations, query sets $Q_s^j, j \in [K] \backslash \mathcal{B}$ are symmetric along all the files $W_i, i \in [N]$. Further, the mapping $\rho$ is a one-to-one map from $[N]$ to $\mathcal{B}$, which ensures that the subfile $W_{i,p_j}$ is associated with each $j' \in \mathcal{B}$. Hence databases will not be able to determine which two users have the same demand.

\end{itemize}

\end{proof}

\section{Comparison of the rate of the proposed scheme with the scheme in \cite{ZWSJC2021}} \label{sec4}

In this section, we compare the rate of the proposed scheme to the rate of PD scheme. The proposed scheme is given for one memory point $M \in (0,N/K)$, which is close to zero. In Figures \ref{figc1} and \ref{figc2}, we compare the rate for different number of databases, files and users, and then we mathematically prove that the memory-rate pair $(M,R)$ of the proposed scheme lies below the lower convex envelope of the points $(0,N)$ and $(tN/K, R_{PD}(tN/K)), t \in \{ 1,2, \ldots K \}$ achieved by PD scheme , where $R_{PD}(tN/K)$ denotes the rate of the PD scheme for $M=\frac{tN}{K}$.

\begin{figure}
\caption{Rate of the proposed scheme compared to the rate of the PD scheme\label{figc2}}
\centering
\includegraphics[scale=0.4]{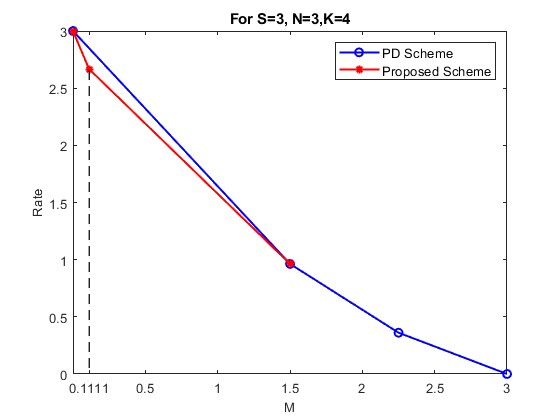}
\includegraphics[scale=0.4]{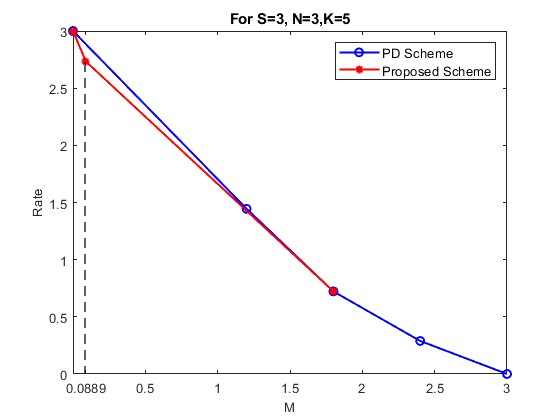}

\includegraphics[scale=0.4]{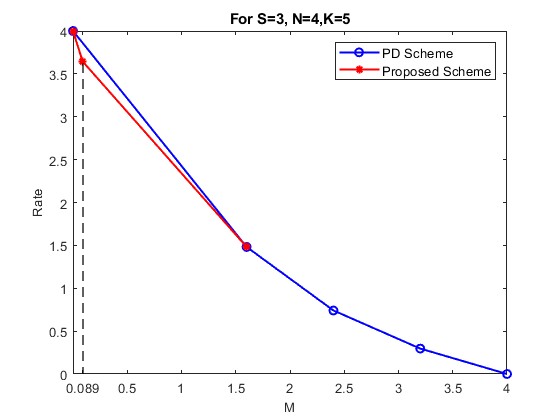}
\end{figure}

\begin{theorem}[\cite{ZWSJC2021}]
The product design achieves the load (rate) of 
$$R_{PD} = \min(N-M, \hat{R}(M))$$
in which 
$\hat{R}(M)= \frac{K-t}{t+1} \left( 1+ \frac{1}{S} + \frac{1}{S^2} + \cdots + \frac{1}{S^{N-1}} \right),$
where $t=\frac{KM}{N} \in [K]$. For non-integer values of $t$, the lower convex envelope of the integer points $(0, R_0)$ and $\left(\frac{tN}{K}, R_{PD}(M) \right)$, $t \in [K]$ are achievable, where $R_0$ is the rate achieved without cache.
\end{theorem}

In our scheme $M \in (0, N/K)$. For comparison, we find the rate $R_{PD}$ of the PD for the value of $M$ proposed in the given scheme. The rate of the PD scheme for $M=tN/K, t \in \{ 1,2, \ldots, K \}$ is 
 \begin{equation}\label{PD}
R_{PD}(tN/K)=\min(N(1- t/K), \hat{R}(tN/K)),
\end{equation}
where $\hat{R}(tN/K)=\frac{K-t}{t+1} \left( 1+ \frac{1}{S} + \frac{1}{S^2} + \cdots + \frac{1}{S^{N-1}} \right)$.

\begin{lemma}\label{lemma4.1}
For the defined integer $q$ in \eqref{q}, following ineqaulity holds
$$NS^{N-1}-q > 0$$
for all $S\geq 2$ and $N\geq 2$.
\end{lemma}

\begin{proof}
Given in Appendix \ref{asec2}.
\end{proof}

\begin{lemma}\label{lemma4.2}
The rate memory pair $(M,R)$ of proposed scheme, where 
$$M=\frac{NS^{N-1}-q}{KS^{N-1}} \quad  \text{and} 
 \quad R=\begin{cases}
\frac{q}{S^{N-1}} & \text{for} \ N=K \\
\frac{N}{K} \left( \frac{q}{S^{N-1}} +K-N \right) & \text{for} \ N<K 
\end{cases},$$
lies below the line passing through the points $(0,N)$ and $\left( \frac{tN}{K}, R_{PD}\left( \frac{tN}{K} \right) \right)$ for each $t \in \{1,2,\ldots, K\}$.
\end{lemma}

\begin{proof}
The equation of line passing through the points $(0, N)$ and $(tN/K, R_{PD}(tN/K)), t \in \{1,2, \ldots, K \}$ (denote $R_{PD}(tN/K)$ by $R_{t}$), is
$$R_{PD}-N = \frac{R_{t}-N}{\frac{tN}{K}-0} (M-0) $$ which can be written as 
\begin{equation}\label{line}
R_{PD} =N - \frac{K}{tN} \left( N-R_{t} \right) M.
\end{equation}
In the proposed scheme, cache size is $M= \frac{NS^{N-1}-q}{KS^{N-1}}$ and rate of the PD scheme for $M$,
\begin{align*}
R_{PD} &= N- \frac{K}{tN} \left( N-R_{t} \right)  \left( \frac{NS^{N-1}-q}{KS^{N-1}} \right) \\
&= N- \frac{1}{t} \left(1-\frac{R_t}{N} \right) \left( N-\frac{q}{S^{N-1}} \right).
\end{align*}

Now we compare the rate of the proposed scheme for both of the following cases.

\begin{itemize}
\item \textbf{For $N=K$:} The rate of the proposed scheme is 
$R=\frac{q}{S^{N-1}}.$ We have
\begin{align*}
R_{PD}-R&=N- \frac{1}{t} \left(1-\frac{R_t}{N} \right) \left( N-\frac{q}{S^{N-1}} \right) -\frac{q}{S^{N-1}} \\
& =\left( 1-\frac{1}{t} \right)\left( N-\frac{q}{S^{N-1}} \right) +\frac{R_t}{t} \left( 1-\frac{q}{NS^{N-1}} \right) \\
&=\left( 1-\frac{1}{t} + \frac{R_t}{tN} \right) \left(N- \frac{q}{S^{N-1}} \right).
\end{align*}
Since $\left( 1-\frac{1}{t} + \frac{R_t}{tN} \right)$ is always a positive number, from Lemma \ref{lemma4.1}, we have $R_{PD}-R > 0$.

\item \textbf{For $N<K$:} The rate of the proposed scheme is 
$R=\frac{N}{K}\left[ \frac{q}{S^{N-1}} +K-N \right].$ We have

\begin{align*}
R_{PD}-R&=N- \frac{1}{t} \left(1-\frac{R_t}{N} \right) \left( N-\frac{q}{S^{N-1}} \right) -\frac{Nq}{KS^{N-1}} -N\left( 1-\frac{N}{K} \right) \\
&= \left(N-\frac{q}{S^{N-1}}\right)  \left( \frac{R_{t}}{tN} + \frac{N}{K} - \frac{1}{t} \right).
\end{align*}
Again from Lemma \ref{lemma4.2}, we have $\left(N-\frac{q}{S^{N-1}}\right)>0$. Therefore to prove that $R_{PD}-R>0$, we need to show that 
$\left( \frac{R_{t}}{tN} + \frac{N}{K} - \frac{1}{t} \right) >0.$
Since $R_t=\min \left(N\left( 1-\frac{t}{K} \right), \hat{R}(tN/K) \right)$, where $\hat{R}(tN/K)=\frac{K-t}{t+1}A$ and $A=\left(1+\frac{1}{S}+\cdots+\frac{1}{S^{N-1}}\right)$.
\begin{enumerate}
\item If $\hat{R}(tN/K) \geq N(1-t/K)$, then $R_{t}=N(1-t/K)$. We have
$$ \frac{R_{t}}{tN} + \frac{N}{K} - \frac{1}{t}   = \frac{N-1}{K} >0.$$

\item If $\hat{R}(tN/K) < N(1-t/K)$, then 
$R_{t}=\hat{R}(tN/K)  = \frac{K-t}{t+1} A $. Clearly, for $K<Nt$, we have $ \frac{R_{t}}{tN} + \frac{N}{K} - \frac{1}{t} >0$. Now for $K>Nt$, we have
\begin{align*}
 \frac{R_{t}}{tN} + \frac{N}{K} - \frac{1}{t} &=  \frac{1}{tN} \left( \frac{K-t}{t+1} A \right) + \frac{N}{K} - \frac{1}{t} \\ 
& >  \frac{K-t}{Nt(t+1)}+ \frac{N}{K} -\frac{1}{t}  \qquad \qquad (\text{as} \ A>1) \\
& = \frac{K(K-t)+N^2 t (t+1)- NK(t+1)}{NKt(t+1)}\\
&= \frac{(K-Nt)^2+ Kt(N-1) - N(K-Nt)}{NKt(t+1)} \\
&= \frac{(K-Nt)^2 - (K-Nt) -(N-1) (K-Nt) + Kt(N-1) }{NKt(t+1)}  \\
&= \frac{(K-Nt) (K-Nt-1) +(N-1) (K(t-1)+Nt) }{NKt(t+1)}  >0  .
\end{align*}
\end{enumerate}

\end{itemize}

\end{proof}

\begin{lemma}\label{lemma4.3}
Suppose $R(M)$ and $R_{PD}(M)$ denote the rates of proposed scheme and the rate of the PD scheme, respectively, at $M=\frac{NS^{N-1}-q}{KS^{N-1}}$, then
$$R_{PD}(M)-R(M)>0.$$
\end{lemma}

\begin{proof}
From Lemma \ref{lemma4.2}, it is clear that the rate memory pair $(M,R)$ of proposed scheme lies below the lower convex envelope of the points $(0,N)$ and $\left( \frac{tN}{K}, R_{PD}\left( \frac{tN}{K} \right) \right), t \in \{1,2,\ldots, K\}$. Hence the proof.
\end{proof}

\section*{Acknowledgment}

This work was supported partly by the Science and Engineering Research Board (SERB) of the Department of Science and Technology (DST), Government of India, through J.C. Bose National Fellowship to Prof. B. Sundar Rajan and by IISc-IoE postdoctoral fellowship awarded to Charul Rajput. The authors would like to thank Mr. Kanishak Vaidya for the valuable discussions and for helping to get an overview of ongoing research on this topic.

\appendices

\section{ Proof of Correctness, Privacy and Achieving Capacity of PIR scheme given in Section \ref{sec2}}\label{asec1}

\begin{lemma}\label{lemma1}
The set $Q_j$ generated by Algorithm \ref{algo1} satisfies the following properties for each $j \in [S]$.
\begin{enumerate}
\item Each subfile $W_i(t), i \in [N], t \in [S^{N-1}]$ appears atmost once in $Q_j$. 

\item Forall $k \in [N]$, $k$-sums of each possible type appears exactly $\phi_j(S,k)$ times in $Q_j$. 

\item Total $\sum_{k=1}^N {N-1 \choose k-1} \phi_j (S,k)$ subfiles of a file $W_i$ appears in $Q_j$.

\end{enumerate}
\end{lemma}

\begin{proof}
\begin{enumerate}
\item In Algorithm $1$, new subfiles of any file $W_i, i \in [N]$ is added to set $Q_j$ only in \eqref{add1} and \eqref{add2}. In \eqref{add2}, the prefix increment operator $++t$ is used to add new subfiles $\{ W_{i_1}(a^{i_1}_{++t_{i_1}}),$ $W_{i_2}(a^{i_2}_{++t_{i_2}}),$ $\ldots, W_{i_k}(a^{i_k}_{++t_{i_k}}) \}$, which ensures that no subfile is added twice in a set $Q_j, j \in [S]$. In \eqref{add1}, the set $\{ Q \cup W_d(a_{++t_d}^d) \}$ is added to $Q_j$, where $Q \in Q_i, i \neq j$. Since only those $Q$ which were added by \eqref{add2} in $Q_i$ will be selected by 'if' condition, therefore, the subfiles in $Q$ are new in $Q_j$. Also the subfile $W_d(a_{++t_d}^d)$ is added using prefix increment operator $++t$, this is also a new subfile in $Q_j$. 

\item Every $k$-sum corresponds to a subset of $[N]$ of size $k$. Let
$\{ U \subseteq [N] \ | \ |U|=k \} = X_1 \cup X_2 ,$
where $X_1=\{U \subseteq [N]  \ | \ |U|=k \ \& \ d \in U\}$ and $X_2=\{U \subseteq [N]\backslash \{d\} \ | \ |U|=k\}.$
From \eqref{add2}, it is clear that all the $k$-sums correspond to the sets in collection $X_2$ are occurring $\phi_j(S,k)$ times in $Q_j$. Now we only need to prove that it is true for the $k$-sums corresponding to the sets in collection $X_1$, which are being added to $Q_j$ in \eqref{add1}. We prove it by induction for $k \in [N]$.
The claim is true for $k=1$ as every $1$-sum appears $\phi_1(S,1)=1$ time in $Q_1$ and $\phi_j(S,1)=0$  time in $Q_j, 2 \leq j \leq S $.
 
Assume that the claim is true for $k=k'-1$, i.e., each $(k'-1)$-sums corresponds to the sets in collection $X_1$ are occurring $\phi_j(S,k'-1)$ times in $Q_j$ for all $j \in [S]$. Now we prove it for $k=k'$.

 In \eqref{add1}, a $k'$-sum $ \{Q \cup W_d(a^d_{++t_d})\}$ which corresponds to a set in collection $X_1$  is being added in $Q_j$, where $Q\in Q_i$ and $|Q|=k'-1$ for all $i \in [S]$ and $ i \neq j $. Since claim is true for all $k=k'-1$, every $(k'-1)$-sum appears exactly $\phi_j(S, k'-1)$ times in $Q_i$. Therefore, the total number of repititions of a $k'$-sum which corresponds to the sets in collection $X_1$ in $Q_j$ is 
$\eta_j=\sum_{i=1, i \neq j}^{S} \phi_i(S,k'-1).$
For $j=1$, 

\begin{align*}
\eta_1&=\sum_{i=2}^{S} \phi_i(S,k'-1) =(S-1) f(S,k'-1) = (S-1) \left[ \frac{1}{S} \left( (-1)^{k'-1} +(S-1)^{k'-2}\right) \right] \\
&=g(S,k)= \phi_1(S,k').
\end{align*}
For $j >1$,
\begin{align*}
\eta_j&=\sum_{i=1, i \neq j}^{S} \phi_i(S,k'-1) =g(S,k'-1)+(S-2)f(S,k'-1) \\
&= \frac{S-1}{S} \left[ (-1)^{k'-2} +(S-1)^{k'-3}\right] + \frac{S-2}{S} \left[ (-1)^{k'-1} +(N-1)^{k'-2}\right] \\
&= \frac{1}{S} \left[ (-1)^{k'} +(N-1)^{k'-1}\right] =f(S,k')=\phi_j(S,k').
\end{align*}

\item Since the total number of repititions of a $k$-sum in $Q_j, j \in [S]$ is $\phi_j(S,k)$ and there are exactly ${N-1 \choose k-1}$ number of $k$-sums which contain a subfile of $W_i$ for a fix $i\in [N]$, the total number of subfiles of a file $W_i$ appears in $Q_j$ is  $\sum_{k=1}^N {N-1 \choose k-1} \phi_j (S,k)$.

\end{enumerate}
\end{proof}

\subsubsection{Correctness of the scheme}
To prove that the scheme described in Algorithm $1$ is correct, we need to show that the user gets all the subfiles of the demanded file $W_d$, where $d \in [N]$. Initially, the subfile $W_d(a_1^d)$ is added to the query set $Q_1$ and any other subfile of the file $W_d$ is being added to $Q_j'$s in \eqref{add1} within the set $\{Q \cup W_d(a^d_{++t_d})\}$, where $Q \in Q_i$ and $i \neq j$. Since query $Q$ is already asked to database $i$, user can easily decode the subfile $ W_d(a^d_{++t_d})$. That means every subfile of $W_d$ that appeared in any of the query set $Q_j, j\in [S]$ can be decoded by the user. From Lemma \ref{lemma1} (a) and (c), we find that the total number of subfiles of $W_d$ appeared in $\cup_{j=1}^{S}Q_j$ is
\begin{align*}
\eta&=\sum_{j=1}^S \sum_{k=1}^N {N-1 \choose k-1} \phi_j(S,k)\\
&=\sum_{k=1}^N {N-1 \choose k-1} g(S,k) + (S-1)\sum_{k=1}^N {N-1 \choose k-1} f(S,k) \\
&= \sum_{k=1}^N {N-1 \choose k-1} \frac{S-1}{S}  \left[ (-1)^{k-1} + (S-1)^{k-2} + (-1)^{k} + (S-1)^{k-1} \right] \\
&= \sum_{k=1}^N {N-1 \choose k-1} (S-1)^{k-1} \\
&= \sum_{k'=0}^{N-1} {N-1 \choose k'} (S-1)^{k'} = (1+S-1)^{N-1} = S^{N-1}.
\end{align*}
Hence all the subfiles of $W_d$ will appear in one of the query sets $Q_j, j \in [S]$ and all are decodable.

\subsubsection{Privacy}
In Algorithm $1$, whenever a subfile $W_i(t)$, where $ i \in [N]$ and $t \in [S^{N-1}]$, is added to the set $Q_j$ for some $j\in [S]$, it is added using a random permutation $a^i$. Form Lemma \ref{lemma1} (a) and (b), it is clear that query sets are symmetric along all the files $W_i, i\in [N]$ as for any $j \in [S]$, no subfile occurs more than once in $Q_j$ and every possible $k$-sum among the files $W_1, W_2, \ldots, W_N$ appears exactly $\phi_j(S,k)$ times in $Q_j$ which is clearly independent of $i$.

\subsubsection{Rate}

Suppose the size of each file $W_i, i\in [N]$ is $L$ bits, then the total number of downloaded bits from all the databases is $\frac{L}{S^{N-1}} |\cup_{j-1}^S Q_j|$, as the answer to each query has the size of a subfile. The total number of required bits is $L$. Therefore, the rate
$$R=\frac{1}{S^{N-1}} \left| \bigcup_{j-1}^S Q_j \right|.$$
The total number of downloaded subfiles is
\begin{align*}
\left| \bigcup_{j-1}^S Q_j \right| &= \sum_{j=1}^S |Q_j| = \sum_{j=1}^S \sum_{k=1}^N {N \choose k} \phi_j(S,k)  = \sum_{k=1}^N {N \choose k} \left[ g(S,k) + (S-1) f(S,k) \right] \\
&= \sum_{k=1}^N {N \choose k} \frac{S-1}{S}  \left[ (-1)^{k-1}  + (S-1)^{k-2} + (-1)^{k} + (S-1)^{k-1} \right] \\
&= \sum_{k=1}^N {N \choose k} \frac{S-1}{S}  \left[S(S-1)^{k-2} \right] 
= \frac{1}{S-1} \sum_{k=1}^N {N \choose k} (S-1)^{k} \\
&= \frac{1}{S-1} \left[ \sum_{k=0}^N {N \choose k} (S-1)^{k} - {N \choose 0} (S-1)^0 \right] \\
&=  \frac{1}{S-1} \left[ (1+S-1)^N - 1 \right] = \frac{S^N-1}{S-1} = 1+S+S^2+\cdots + S^{N-1}.
\end{align*}
Therefore, the rate is $R=\left(1+\frac{1}{S}+\frac{1}{S^2}+ \cdots + \frac{1}{S^{N-1}} \right) ,$
which is optimal.

\section{Proof of Lemma \ref{lemma4.1}} \label{asec2}
From the definition of functions $\phi_s(S,k)$ and $\psi_s(S,k)$, we have
\begin{enumerate}
\item For $k=1$ and $2 \leq s \leq S$,
$\phi_s(S,k)=f(S,1)=\frac{1}{S} \left( (-1)^1 + (S-1)^0 \right) =0.$
Therefore, $\psi_s(S,1)=0$ for all  $2 \leq s \leq S$. .
\item For $k=2$ and $s=1$,
$\phi_s(S,k)=g( S,2)=\frac{S-1}{S} \left( (-1)^1 + (S-1)^0 \right)=0.$
Therefore, $\psi_1(S,2)=0$.
\item For $k=N$, $\psi_s(S,k)=(N-1)\phi_s(S,k)$ is an integer as $\phi_s(S,k)$ is an integer.
\end{enumerate}
We know that $\lceil x \rceil \leq x+1$ and if $x$ is an integer then $\lceil x \rceil = x$. Hence from the observations 1), 2), and 3) above, we have  
\begin{align*}
q & = \sum_{k=1}^{N} \sum_{s=1}^{S} {N \choose k} \left \lceil \frac{k(N-1)}{N} \phi_s(S,k) \right \rceil \\
&<  \sum_{k=1}^{N} \sum_{s=1}^{S} {N \choose k} \left ( \frac{k(N-1)}{N} \phi_s(S,k) +1 \right ) - \sum_{s=2}^S {N \choose 1} - {N \choose 2} - \sum_{s=1}^S {N \choose N} \\
&=  \sum_{k=1}^{N} \sum_{s=1}^{S}   \frac{k(N-1)}{N} {N \choose k} \phi_s(S,k)   + \sum_{k=1}^{N} \sum_{s=1}^{S} {N \choose k}  - N(S-1) - {N \choose 2} - S \\
&= (N-1) \sum_{k=1}^{N} \sum_{s=1}^{S}  {N-1 \choose{k-1}} \phi_s(S,k)  + S \sum_{k=1}^{N} {N \choose k}  - N(S-1) - {N \choose 2} - S \\
&= (N-1) S^{N-1} + S (2^N-1)  - N(S-1) - {N \choose 2} - S .
\end{align*}
Therefore,
\begin{align}\label{q3}
NS^{N-1}-q &>  
 S( S^{N-2} -  2^N+2)  + N(S-1) + {N \choose 2}.
\end{align}
\begin{itemize}
\item For $S \geq 4$, we have
\begin{align}\label{q4}
NS^{N-1}-q &>  4( 4^{N-2} -  2^N+2)  + 3N + {N \choose 2} =4(2^{2N-4}-2^N+2)+ 3N + {N \choose 2}.
\end{align}
For $N \geq 4$, we have $2N-4 \geq N$. Using this inequality in \eqref{q4}, we have
$
NS^{N-1}-q >  4(2^N-2^N+2)+ 3N + {N \choose 2} = 8+3N+{N \choose 2} >0.
$
For $N=3$, using \eqref{q4}, we have
$NS^{N-1}-q > 4(2^2-2^3+2)+9+3 =4 >0.$

\item For $S=3$,  we have
\begin{equation}\label{q7}
NS^{N-1}-q > 3(3^{N-2}-2^N+2) +2N+ {N \choose 2}.
\end{equation}
If $N>5$, then we have $3^{N-2}>2^N$, and for $N=5$, it is easy to check that the right-hand side of \eqref{q7} is a positive number. Hence, for $N\geq 5$, we have $NS^{N-1}-q >0$. For $N=3,4$, the value of $NS^{N-1}-q $ is positive, which satisfies the claim.
\end{itemize}

Therefore, for all $S\geq 3$ and $N \geq 3$, we have $NS^{N-1}-q > 0$. Now we separately prove it for the remaining cases $S=2$ and $N=2$ which will cover all the cases when $S \geq 2$ and $N \geq 2$.


\subsection{For $S=2$:} By the definition of $g(S,k)$ and $f(S,k)$, we have
$$
g(2,k) =
 \begin{cases} 1 & \text{if $k$ is odd} \\ 0 & \text{if $k$ is even} \end{cases}
\ \
\text{and}
\ \
f(2,k) = 
\begin{cases} 0 & \text{if $k$ is odd} \\ 1 & \text{if $k$ is even} \end{cases}.
$$
Therefore, we have $G(2,k)=0$ if $g(2,k)=0$ and if $g(2,k)=1$ then
$
G(2,k)=\left \lceil \frac{k(N-1)}{N} \right \rceil = k -\left \lceil \frac{k}{N} \right \rceil +1 = k$ for $k <N$
and $G(2,k) =N-1$ for $k=N$. Therefore, we have
\begin{equation}\label{q5}
G(2,k)= \begin{cases} 0 & \text{if $k$ is even} \\ k & \text{if $k$ is odd \& $k<N$} \\ N-1 & \text{if $k$ is odd \& $k=N$} \end{cases}.
\end{equation}
Similarly,
\begin{equation}\label{q6}
F(2,k)= \begin{cases} 0 & \text{if $k$ is odd} \\ k & \text{if $k$ is even \& $k<N$} \\ N-1 & \text{if $k$ is even \& $k=N$} \end{cases}.
\end{equation}
If $N$ is even, then the value of $q$ can be calculated as follows.
\begin{align*}
q &= \sum_{k=1}^N {N \choose k} [G(S, k) + (S-1) F(S,k)] \\
&=  \sum_{k=1}^N {N \choose k} G(S, k) + (S-1)  \sum_{k=1}^N {N \choose k} F(S,k) \\
&= \left[{N \choose 1} + 3 {N \choose 3} + 5 {N \choose 5} + \ldots  + {N \choose{N-1}}  (N-1) \right] + (2-1) \left[2 {N \choose 2} \right. \\
& \qquad  \qquad \left. + 4 {N \choose 4}  + \ldots  + {N \choose N}(N-1) \right] \\
& = \left[{N \choose 1} + 2 {N \choose 2} + 3 {N \choose 3} + \ldots + {N \choose{N-1}}   (N-1)  + {N \choose N} (N-1) \right]\\
&=\sum_{k=1}^N k {N \choose k} - {N \choose N} = N2^{N-1}-1.
\end{align*}
Similarly, for the case when $N$ is odd, we have $q=N2^{N-1}-1$. Therefore, we have 
$NS^{N-1}-q=1>0.$

\subsection{For $N=2$:} By the definitions of $g(S,k)$ and $f(S,k)$, we have
$$g(S,k) = \begin{cases} 1 & \text{for} \ k=1 \\ 0 & \text{for}\ k=2 \end{cases} \ \
\text{and} 
\ \  f(S,k) = \begin{cases} 0 & \text{for} \ k=1 \\ 1 & \text{for}\ k=2 \end{cases}.$$
Further, we have $G(S,1)=1, G(S,2)=0, F(S,1)=0, F(S,2)=1$, and 
$$q=\sum_{k=1}^2 {2\choose k} (G(S,k) + (S-1) F(S,k)) = S+1.$$
Therefore, $NS^{N-1}-q=2S-(S+1)=S-1>0$.


\ifCLASSOPTIONcaptionsoff
  \newpage
\fi

\end{document}